\def\supp{0}
\def\arxiv{1}
\ifnum\arxiv=0
\documentclass[lettersize,journal]{IEEEtran}

\usepackage{amsmath,amsfonts,amssymb,amsthm}
\usepackage{array}
\usepackage[caption=false,font=normalsize,labelfont=sf,textfont=sf]{subfig}
\usepackage{textcomp}
\usepackage{stfloats}
\usepackage{url}
\usepackage{verbatim}
\usepackage{graphicx}
\usepackage{cite}
\usepackage{tcolorbox}
\usepackage{mathtools}
\usepackage{enumitem}
\usepackage{relsize}
\usepackage{soul}
\usepackage{color}
\usepackage{thmtools}
\usepackage{thm-restate}
\usepackage{hyperref}
\usepackage[capitalise]{cleveref}
\usepackage{algorithm}
\usepackage{algpseudocode}
\usepackage{xspace}
\usepackage{tikz}

\newtheorem{theorem}{Theorem}[section]
\newtheorem{lemma}[theorem]{Lemma}
\newtheorem{corollary}[theorem]{Corollary}
\newtheorem{definition}[theorem]{Definition}
\newtheorem{assumption}{Assumption}
\newtheorem{example}{Example}
\newtheorem{remark}{Remark}

\crefname{assumption}{Assumption}{Assumptions}
\else
\def\supp{0} 

\documentclass[11pt]{article}

\usepackage[a4paper,margin=1in]{geometry}
\usepackage[T1]{fontenc}
\usepackage[utf8]{inputenc}

\usepackage{amsmath,amsfonts,amssymb,amsthm}
\usepackage{array}
\usepackage{textcomp}
\usepackage{url}
\usepackage{verbatim}
\usepackage{graphicx}
\usepackage{cite}
\usepackage{tcolorbox}
\usepackage{mathtools}
\usepackage{enumitem}
\usepackage{relsize}
\usepackage{soul}
\usepackage{xcolor}
\usepackage{thmtools}
\usepackage{thm-restate}
\usepackage{hyperref}
\usepackage[capitalise]{cleveref}
\usepackage{algorithm}
\usepackage{algpseudocode}
\usepackage{xspace}
\usepackage{tikz}
\usepackage{subcaption}

\newtheorem{theorem}{Theorem}[section]
\newtheorem{lemma}[theorem]{Lemma}
\newtheorem{corollary}[theorem]{Corollary}
\newtheorem{definition}[theorem]{Definition}

\newtheorem{remark}{Remark}

\crefname{assumption}{Assumption}{Assumptions}

\title{Metastability-Containing Turing Machines}

\author{
Johannes Bund \and
Amir Leshem \and
Moti Medina
\thanks{All authors are with the Faculty of Engineering, Bar-Ilan University, Ramat Gan, Israel.
Johannes Bund was also affiliated with Universit\'{e} Paris-Saclay, CNRS, ENS Paris-Saclay, Laboratoire M\'{e}thodes Formelles, France, when parts of this work were carried out.
E-mail: bundjoh@biu.ac.il, amir.leshem@biu.ac.il, moti.medina@biu.ac.il.
Funding: This research was supported by the Israel Science Foundation under Grants 867/19 and 554/23. Amir Leshem was partially supported by ISF Grant 2197/22.}
}

\date{}
\fi

\newcommand{\moti}[1]{\hl{\textbf{M}: #1}}

\newcommand{\joh}[1]{\hl{\textbf{J}: #1}}

\newcommand{\al}[1]{\hl{\textbf{A}: #1}}

\DeclareMathOperator{\res}{res}

\newcommand{\IN}{\mathbb{N}}
\newcommand{\IR}{\mathbb{R}}
\newcommand{\IB}{\{0,1\}}
\newcommand{\IT}{\mathbb{T}}
\newcommand{\AND}{\textsf{and}}
\newcommand{\OR}{\textsf{or}}
\newcommand{\bin}[2]{\textrm{bin}_{#1}(#2)}
\newcommand{\NOT}{\textsf{not}}

\DeclareMathOperator{\eqdef}{\coloneqq}

\DeclareMathOperator{\NMUX}{(n:1)-\MUX}
\DeclareMathOperator{\NCMUX}{(n:1)-\CMUX}
\newcommand{\orcid}[1]{\href{https://orcid.org/#1}{ {\includegraphics[scale=0.5]{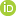}}}}
\newcommand{\mfu}{\ensuremath{\mathfrak{u}}}
\newcommand{\starr}{*}
\newcommand{\bigstarr}{\mathop{\raisebox{-.7pt}{\ensuremath{\mathlarger{\mathlarger{\mathlarger{*}}}}}}}

\newcommand{\blank}{\textnormal{\texttt{B}}}%
\newcommand{\leftt}{\textnormal{\texttt{L}}}%
\newcommand{\rightt}{\textnormal{\texttt{R}}}%
\newcommand{\tape}{\textnormal{\texttt{TAPE}}}

\newcommand{\NP}{\textsf{NP}}

\newcommand{\BO}{\mathcal{O}}

\newcommand{\ttime}{\ensuremath{{\mathsf{TIME}}}\xspace}
\newcommand{\exptime}{\ensuremath{{\mathsf{EXPTIME}}}\xspace}
\newcommand{\CONP}{\ensuremath{{\mathsf{co\text{-}NP}}}\xspace}
\newcommand{\PP}{\ensuremath{{\mathsf{P}}}\xspace}

\newcommand{\IRS}{image resolution set}

\newcommand{\expMC}{\text{{\sc Exp-\MC-Bhp}}\xspace}

\newcommand{\sizeMC}{\text{{\sc \mfu-detect}$_1$}\xspace}
\newcommand{\sizeMCm}{\text{{\sc \mfu-detect}$_m$}\xspace}
\newcommand{\expbhp}{\text{{\sc BHP}}\xspace}
\newcommand{\pexpbhp}{\text{{\sc PExp-BHP}}\xspace}

\newcommand{\Tautology}{\text{{\sc Tautology}}\xspace}
\newcommand{\sizearbMC}{\text{{\sc \mfu-detect}$^{\poly}_{\geq0}$}\xspace}

\newcommand{\MC}{\textsc{MC}\xspace}

\newcommand{\FMC}{{\cal M}}

\newcommand{\TM}{\textsc{TM}\xspace}
\newcommand{\eval}{\textsc{Eval}\xspace}
\newcommand{\UTM}{\textsc{UTM}\xspace}

\newcommand{\UMCTM}{{\cal U}_{\mfu}}
\newcommand{\MUX}{\textsf{MUX}\xspace}
\newcommand{\CMUX}{\textsf{CMUX}\xspace}

\newcommand{\poly}{\textrm{poly}}

\newcommand{\NO}{\textsc{NO}}

\newcommand{\return}{\textbf{return }}
\newcommand{\IIf}{\textbf{If }}
\newcommand{\then}{\textbf{then }}
\newcommand{\otherwise}{\textbf{otherwise }}
\newcommand{\reject}{\textit{reject }}
\newcommand{\accept}{\textit{accept }}

\begin{document}

\ifnum\arxiv=0
\title{Metastability-Containing Turing Machines}
\author{%
\IEEEauthorblockN{Johannes Bund\IEEEauthorrefmark{1}$^{ \orcid{0000-0002-1108-1091}}$, \and Amir Leshem\IEEEauthorrefmark{1}$^{ \orcid{0000-0002-2265-7463}}$, {\textit{ Fellow, IEEE}}, \and Moti Medina\IEEEauthorrefmark{1}}$^{ \orcid{0000-0002-5572-3754}}$\\\vspace{1em}%
\IEEEauthorblockA{\IEEEauthorrefmark{1}\textit{Faculty of Engineering, Bar-Ilan University}\\
Ramat Gan, Israel \\
bundjoh@biu.ac.il, amir.leshem@biu.ac.il, moti.medina@biu.ac.il}
\thanks{This work has been done while Johanes Bund was at Bar-Ilan University and Universit\'{e} Paris-Saclay, CNRS, ENS Paris-Saclay, Laboratoire M\'{e}thodes Formelles. 
\emph{Funding:} {
{This research was supported by the Israel Science Foundation under Grant
867/19 and 554/23. Amir Leshem was partially supported by ISF grant 2197/22.}}}
}
\fi

\maketitle

\begin{abstract}
  Metastability is a spurious mode of operation in digital signals, where an electrical signal fails to settle into a stable state within a specified time, leading to uncertainty and potentially failing downstream hardware. 
A system that computes the closure over all possibilities, given an uncertain input, is called a Metastability-containing system. 

While prior work has addressed metastability-containing systems in the context of combinational and clocked circuits, state machines, and logic formulas, its implications for general-purpose computation remain largely unexplored. 

In this work, we study the metastability-containing systems within an abstract computational model: The Turing Machine. 
This approach allows us to investigate the computational limits and capabilities of Turing Machines operating under uncertain inputs. Specifically, we prove that in general the metastable closure of a Turing Machine is non-computable. Then we discuss cases where the meta-stable closure is computable: For \exptime problems, we prove that resolving even a single uncertain bit is \exptime-complete. In contrast, we prove that for polynomial time problems, the meta-stable closure is polynomial time computable for a logarithmic number of uncertain bits, but \CONP-complete, when the number of undefined inputs is arbitrary. Finally, we describe a hardware-realizable Universal Turning Machine that computes the metastable closure of any given bounded-time Turing Machine with at most an exponential blowup in time.  
\end{abstract}

\ifnum\arxiv=0
\else
\fi

\section{Introduction}\label{sec:intro}
Metastability is an uncertain equilibrium state that is inherent to any machine that operates 
    on bistable storage elements, e.g., in electronic circuits, registers can become metastable 
    when the inputs have incorrect timing~\cite{ginosar11tutorial}. 
In turn, high-speed microprocessors that receive many inputs from sensors, 
    are likely to see metastable inputs. 
The metastable state (denoted by $\mfu$) will eventually resolve to a stable state ($0$ or $1$).
Since the stable state and resolution time are not known a priori, following~\cite{friedrichs18containing} we consider the metastable 
    state to be a superposition of the stable states that propagates through the circuits. 
    
Combinational circuits that have the best possible behavior under metastable inputs are called 
    \emph{metastability-containing} (M-Containing) (we elaborate on this in 
     Section~\ref{sec:closure} and Appendix~A\ifnum \supp=1
       \text{ }in 
      the supplementary material\fi). 
It is well known that for each combinational circuit, there is an equivalent M-containing circuit that may differ in 
    circuit size and depth (see, e.g.~\cite{huffman57design}). 
More recently, M-containing circuits gained interest again, when an unconditional lower bound on their complexity was proven in~\cite{IkenmeyerKLLMS19}. 
Intuitively, the result showed that there are circuits where the smallest M-Containing circuit 
    is exponentially larger than the smallest non-M-Containing circuit.
Naturally, the interest in studying M-Containment in other computational models emerges.
M-Containment in Finite State Machines was discussed in~\cite{BundLM22} and~\cite{TarawnehFL17}.

\sloppy
The most prevalent computational model in 
    Computer Science is the \emph{Turing machine} (\TM)~\cite{AB09}.
A \TM is essentially a Finite State Machine that has an unlimited
    tape on which it can store and access data.
It is well known that \TM's are equally expressive as any other proposed model of computation
    such as $\lambda$-calculus, register machines, cellular automata, Conway's game of life, and others~\cite{AB09}.
As such, \TM's seem able to simulate all physically realizable models of computation with very little loss of efficiency. The well-known Church-Turing thesis \cite{turing1936computable,church1936unsolvable} asserts that any computation can be simulated by an equivalent Turing machine. 
Hence, the study of \TM's makes it possible to faithfully compare the complexity of algorithms, and argue about their efficiency.

\subsection{Related Work}
M-Containing combinational circuits are a class of ``fault-tolerant''  combinational circuits~\cite{friedrichs18containing}.~\footnote{When it is clear from the context, we write ``circuits'' instead of ``combinational circuits''. }
In a model where inputs to a combinational circuit may become faulty (metastable) 
    such that the input to the circuit is (partially) unknown, a M-Containing circuit 
    reduces the uncertainty in its outputs due to unknown inputs to a minimum.
In other words, an M-Containing circuit computes the most precise output possible~\cite{IkenmeyerKLLMS19}.

In this work, we extend the notion of M-Containment to \TM's and formally define
    how \TM's operate with metastable inputs.
To our knowledge, there are no results on faults in \TM's, due to metastability. 

Possibly the most related results to this paper are given in~\cite{IkenmeyerKLLMS19} and~\cite{KomarathS20}.
In~\cite[Coro.~1.5]{IkenmeyerKLLMS19}, the authors regard a family of circuits, which is a computational 
    model that is stronger than simple combinational circuits.
In~\cite{KomarathS20}, the authors show that there is no algorithm detecting violations of 
    M-Containment in polynomial time assuming that the strong exponential time hypothesis is true.
To our knowledge the study of M-Containing \TM's has not 
    been investigated in the literature so far.
\subsection{Novelty and Contribution}
%
%
This paper explores \TM's under metastable inputs.
In this work, we lift the theory from combinational circuits to general-purpose computers, extending the theory of~\cite{friedrichs18containing}.
We provide new results regarding the computability of the meta-stable closure for general computable functions, then we provide complexity results and describe a detailed construction of meta-stable containing Turing machines using Kleene Ternary logic.

\ifnum\arxiv=0
{\bf \em Computability of the Metastable Closure. }
\else
\paragraph{Computability of the Metastable Closure. }
\fi
In Section~\ref{sec:computability}, we establish the noncomputability of the metastable closure. Given an arbitrary Turing machine $M$ and a ternary input, the computational task $\FMC$ of determining the size of the set of values obtained by applying the function computed by $M$ to the resolution set of the ternary input, which we refer to as \emph{\IRS}, is noncomputable. This is formalized in the following theorem. 
\begin{restatable}{theorem}{notcomp}\label{thm:notcomp}
    The function $\FMC$ is non-computable.      
\end{restatable}
\ifnum\arxiv=0
{\bf \em Deciding whether there is a single resolution of a single uncertain bit in a Turing machine running a finite time is \exptime-Complete. }
\else
\paragraph{Deciding whether there is a single resolution of a single uncertain bit in a Turing machine running a finite time is \exptime-Complete. }
\fi
While the metastable closure is non-computable, we can focus on computable approximations: Assume that the input of a Turing machine contains metastable bits. What is the size of the \IRS\ of that Turing machine after $k$ computation steps, where the input is the machine, its metastable inputs, and the number of execution steps? 
In Section~\ref{sec:exptime}, we prove that the problem above is hard. Moreover, we prove that resolving a single metastable bit is exponentially hard. This is a sharp contrast to the naive expectation that MC Turing machines should be complex because they need to resolve exponentially many input instances when there are many unresolved inputs, and resolving few instances should be easy. To that end, we define a language $L(\sizeMC)$ and prove that it is \exptime-complete:
\begin{restatable}{theorem}{exphard}\label{thm:exphard}
    The language $L(\sizeMC)$ is \exptime-complete. 
\end{restatable}
\ifnum\arxiv=0
{\bf \em Deciding whether there is a single resolution of any number of uncertain bits in a Turing machine running in polynomial time is \CONP-Complete. }
\else
\paragraph{Deciding whether there is a single resolution of any number of uncertain bits in a Turing machine running in polynomial time is \CONP-Complete. }
\fi
We have already proved that the size of the \IRS\ is non-computable, and that even resolving a single bit can be exponentially hard. However, there is yet another twist. If we limit ourselves to polynomial-time problems (P). Resolving a small (logarithmic) number of bits is easy and can be done in polynomial time. However, resolving any number of input bits is once again hard. More precisely, it is a \CONP -complete problem. In section \ref{sec:coNP} we prove 
\begin{restatable}{theorem}{coNP}\label{thm:coNP}
    The language $L(\sizearbMC)$ is \CONP-complete. 
\end{restatable}
\ifnum\arxiv=0
{\bf \em Kleene Logic and Combinational Circuits. }
\else
\paragraph{Kleene Logic and Combinational Circuits. }
\fi
In Section~\ref{sec:kleene}, we define a subclass of Turing Machines, which we refer to by \emph{Natural \TM's} (see Definition~\ref{def:nattm}), where their transition functions are restricted to functions that are implementable by combinational circuits in Kleene logic. Here, we also briefly explain why not all ternary functions are implementable by combinational circuits and show that one cannot detect or resolve a metastable bit using combinational circuits. In turn, this class of Natural \TM's is also implementable by combinational circuits and standard memory devices. 

\ifnum\arxiv=0
{\bf \em Implementation of a Natural and Universal M-Containing Turing machine. } 
\else
\paragraph{Implementation of a Natural and Universal M-Containing Turing machine. }
\fi
Finally, in Section~\ref{sec:simulation} (complementing the \exptime-completeness discussion of Section~\ref{sec:exptime}), we present a Universal Natural M-Containing Turing machine. This machine simulates the metastable closure of any given (computable) Turing Machine. The simulation comprises an exponential time overhead, i.e., the exponential simulation will result in an \exptime\ running time for any \exptime\ problem. 

%
%
%
%
%
\begin{restatable}{theorem}{mcuni}\label{thm:mcuni}
    Given a \TM $M$ of running time $t(n)$ and input $x\in\IT^*$, then $\UMCTM$ computes $M_\mfu(x)$ with exponential blowup $\tilde{\BO}(t(n)\cdot{2^n})$.
\end{restatable}

The $\UMCTM$ is a Universal, Natural and Oblivious \TM, i.e., it can be implemented by standard gates, and the movements of the reading head depend only on the length of the input and not its content. 

In Corollary~\ref{cor:generality} we extend Theorem~\ref{thm:mcuni} to compute the \MC of arbitrary \TM's with unknown running time. The running time is bounded by an arbitrary bound $T$ on the computational steps.
Furthermore, the extension allows for additional uncertainty in the description of the \TM and $T$. 
This proves that any computability class containing \exptime is closed under the metastable closure operator.

\subsection{Further Related Work}\label{sec:related}
Before describing the structure of the paper, we also provide other related works that give the general context for our results. 
In many cases where the input to an algorithm is partially uncertain, it is possible to deduce parts of the complete output. For example, the \emph{multiplexer} is a circuit that selects one of two inputs: if $s=0$, then it outputs the value input to its $a$ input, and if $s=1$m, then it outputs the value input to its $b$ input. When the select signal $s$ is metastable, the standard multiplexer fails to select the correct input even if $a=b=1$~\cite{FriedrichsK17}.
An M-Containing multiplexer solves this problem; it will output a $1$, e.g., see the \CMUX\ circuit given in Section~\ref{sec:simcmux}. The \CMUX\ circuit has two additional gates compared to the classical non-M-Containing multiplexer. 

Huffman noted that although every combinational circuit has an M-Containing counterpart, the attained M-Containing circuits are exponentially larger than the combinational counterpart~\cite{huffman57design}.
Recently, the super polynomial lower bounds for the size of combinational M-Containing circuits were established by~\cite{IkenmeyerKLLMS19}. 
The result gained interest because it connected M-Containing circuits to monotone circuits. An improved upper bound was given by Jukna~\cite{Jukna21}.
Friedrichs et al.~\cite{friedrichs18containing} establish that synchronous circuits suffer the same blow-up as combinational circuits. Friedrichs et al.~\cite{friedrichs18containing} use the fact that in digital circuit design, it is well established that synchronous and combinational circuits have the same computational power.
Synchronous circuits, i.e., circuits that use a clock and registers to compute an output in multiple rounds, can be transformed into circuits that compute the same output in one round but require more space, i.e., trading off circuit size with multiple computation rounds in a clocked circuit.
This technique is called \emph{unrolling} of the circuit. 
Despite the general lower bound, specific tasks can have optimal size and depth circuits performing on par with the Non-M-Containing circuit~\cite{BundLM17,bund2019optimal}.

No sub-exponential time algorithm can detect whether a circuit is M-Containing, assuming the strong exponential time hypothesis is true~\cite{BN20}. 
In~\cite{IkenmeyerKS23}, the authors consider M-Containing formulas, i.e., fan-out one combinational circuits. 
By extension of Krachmer-Widgerson games to M-Containing formulas, the authors present lower bounds on the complexity of M-Containing formulas. In both works, \cite{IkenmeyerKLLMS19} and~\cite{IkenmeyerKS23}, it is emphasized that the complexity of M-Containing circuits could help to link Boolean circuit complexity to monotone circuit complexity.

A topological argument by Marino~\cite{Marino81} shows that the metastable state cannot be avoided in bistable devices, such as registers in digital circuits.
    
Historically, the problem of unknown inputs was studied early on by Goto~\cite{goto49relay} and has been discussed in various forms throughout the years~\cite{huffman57design,unger1995hazards}.
Existing research in cybersecurity~\cite{hu12complexity,tiwari09flow} seem 
    comes up with the same concept in their own terms.
This is evidence that studying M-Containment is important in circuits and beyond.

The model of M-Containing circuits is realistic, in the sense 
that it meets the real behavior of     CMOS\footnote{Complementary Metal–Oxide–Semiconductor} gates implemented in silicon~\cite{BundLM20}. Theory of M-Containing circuits has been used to design fault-tolerant CMOS circuits~\cite{FuggerKLP17,FriedrichsK17}.

Metastable signals are a common issue in communication between clock domains; to reduce risks, messages are delayed.
The role of M-Containing circuits for clock synchronization and clock domain crossing is recognized in~\cite{BundFLM20,BundFLMR20,BundFM23}.
The authors present high-speed, low-latency solutions for M-Containing synchronization of clock domains.
    
The concept of metastability contacting circuits appears under different names in the literature. 
Most prominently, Huffman~\cite{huffman57design} coined the term \emph{static hazards}. 
In our model and terminology, circuits free of static hazards and M-Containing circuits are synonyms. 
A deeper discussion can be found in~\cite{bund2022hazard}.
%
\ifnum\arxiv=0
\subsection*{Organization of the paper}
\else
\paragraph{Organization of the paper. }
\fi

After the introduction in \cref{sec:intro}, we give the necessary definitions of \TM's, computability, and complexity in \cref{sec:prelim}.
Preliminaries on metastable computation are presented in \cref{sec:closure}, which introduces the \emph{Metastable Closure}.
\Cref{sec:computability} states our first result, non-computability of the metastable closure.
By focusing our attention on the metastable closure of computable functions in \cref{sec:exptime}, we demonstrate \exptime-hardness.
In \cref{sec:kleene}, we introduce the Kleene logic to \TM's and define the class of Natural \TM's that are amenable to hardware implementation.
Finally, in \cref{sec:simulation}, we design and prove correct a Natural Universal \TM.

\section{Preliminaries}\label{sec:prelim}
In this section, we set up the notation used throughout the paper and describe the preliminary material required for the main proofs.

\ifnum\arxiv=0
\medskip
\noindent
\textbf{Notations. }
\else
\paragraph{Notations. }
\fi
Let $A$ denote a set, then $A^* \coloneqq \cup_{i\in \IN} A^i$, where for $i \in \IN$, $A^i \coloneqq \{a_1 \circ \ldots \circ a_i \mid \forall j:a_j \in A\}$, where $\circ$ is the concatenation operator. Note that $A^0$ is defined as the ``empty string''.
%
Let $y \in \IB^n$; we denote by $y_i$ the $i$th bit of $y$ and by $y[i:j]$ for $0 \leq i \leq j\leq n-1$ the substring obtained by removing the leftmost $i$ bits and the rightmost $n-1-j$ bits. 

\subsection{Turing Machines}\label{sec:tm}
Turing Machines are a traditional tool to study computability and complexity
    of computational tasks.
We use the following definition of a Turing machine (\TM)
    given by Hopcroft~\cite{HopcroftU79} for deterministic \TM's with a single, unlimited tape:
\begin{definition}[Turing Machine]\label{def:dettm}
  A \emph{Turing Machine (\TM)} is a $7$-tuple, $(Q, \Sigma, \Gamma,\delta, q_0, \blank, F)$,
    where $Q$, $\Sigma$, $\Gamma$ are all finite sets and
  \begin{itemize}
      \item $Q$ is the set of states,
      \item $\Gamma$ is the tape alphabet,
      \item $\blank$ is the blank symbol and $\blank\in\Gamma$,
      \item $\Sigma$ is the input alphabet and $\Sigma\subseteq\Gamma\setminus\{\blank\}$,
      \item $\delta\colon Q\times\Gamma\rightarrow Q\times\Gamma\times\{\leftt,\rightt\}$ is the transition function,
      \item $q_0\in Q$ is the initial state,
      \item $F\subseteq Q$ is the set of final states. 
  \end{itemize}
\end{definition}
The input to a \TM is a finite sequence of symbols from the input alphabet.
Initially, the \TM receives its input string on the tape, followed by an infinite number of $\blank$ symbols. 
An \emph{execution} of a \TM on a specific input is performed in discrete \emph{steps}.  
In each step of an execution, the \TM reads a symbol from the tape at the position of the tape-head.
According to this symbol and the current state of the \TM, the transition function determines:
\begin{itemize}
  \item the next state, 
  \item the symbol that is written to the tape at the position of the tape-head, and
  \item whether to move the tape-head to the left or the right w.r.t. its current location. 
\end{itemize}
Every tape position that is not written by the \TM remains as it was at the beginning of the step.
As soon as the \TM reaches one of the final states in $F$, it \emph{terminates};  
    no further execution steps are performed. 
A \TM that executes infinitely many steps on input $x$ and never reaches a state in $F$ is \emph{non-terminating on $x$}.  

Let $M$ be a \TM that terminates on every input,
    then the \emph{time complexity} of $M$ is given by a function $T\colon\IN\rightarrow\IN$, 
    where $t(n)$ is the maximum number of execution steps that $M$ performs on any input of length $n$. 

We say a \TM computes the value $y\in (\Gamma\setminus\{\blank\})^*$ if it terminates with only $y$ on its tape.
\begin{definition}[Computable Functions]\label{def:compfunc}
    A function $f\colon\Sigma^*\rightarrow(\Gamma\setminus\{\blank\})^*$ is a \emph{computable function} if
      there is a \TM $M$, which terminates on every input $w\in\Sigma^*$ with just $f(w)$ on its tape.
    We say $M$ computes function $f$ and denote the output of $M$ applied on $x$ by $M(x)$. 
\end{definition}

When restricting the set of functions to $f\colon\{0,1\}^*\rightarrow\{0,1\}$, i.e., \emph{Boolean} functions with a single output bit, we identify $f$ with the set of strings in $\{0,1\}^*$ for which the function $f$ returns $1$, i.e., $L_f \coloneqq \{x \in \{0,1\}^* : f(x)=1\}$. 
We refer to such $L_f$'s as \emph{languages} or \emph{decision problems}. 
%
In this context, given an input string, a \TM \emph{accepts}, \emph{rejects}, or \emph{loops}, if it terminates and outputs a $1$ (True), respectively $0$ (False), on its tape, or does not halt. 
The \emph{language} \emph{recognized} by a \TM $M$, denoted by $L(M)$, is the collection of all strings accepted by $M$. For brevity, we refer to the language recognized by $M$, simply by the \emph{language of $M$}. 
%
%
Languages recognized by a \TM that terminates on \emph{every} input are \emph{decidable}. 

\begin{definition}[Recognizable and Decidable Languages]
    A language $\mathcal{L}$ is \emph{recognizable} if there is a \TM $M$ that recognizes it and $\mathcal{L}=L(M)$. A language $\mathcal{L}$ is \emph{decidable} if there is a \TM $M$ that terminates on every input and $\mathcal{L}=L(M)$.
\end{definition}
Decidable languages can be classified by the time complexity of a \TM deciding them in minimal time.
Denote by $\ttime$ the family of time complexity classes that classifies languages by their time requirements.
We make use of the definition given by Sipser~\cite{Sipser97}.~\footnote{In this paper we discuss \emph{determinsitc} \TM's~\ref{def:dettm}. Hence, time complexity and complexity classes are w.r.t.\ to deterministic machines, i.e., deterministic time~\cite[Def.~1.12]{AB09}. }
\begin{definition}[Time Complexity]
    Let $t\colon \IN \rightarrow \IR_{>0}$ be a function. 
    The time complexity class $\ttime(t(n))$ is defined to be the collection of all languages that are 
        decidable by a Turing machine with time complexity $\BO(t(n))$.    
\end{definition}

In the following paragraphs we define three types of \TM's; the Universal \TM, 
    the Universal \TM with a time-bound, and the Oblivious \TM.
\ifnum\arxiv=0
\paragraph*{Universal Turing Machines}
\else
\paragraph{Universal Turing Machines. }
\fi
A \TM that can simulate any other \TM $M$ on any given input $w$ is called a \emph{Universal \TM
 } (\UTM).
Such a \UTM receives an encoding of $M$ and $w$ as an input, denoted by $\langle M,w \rangle$.
It can simulate an execution of $M$ on $w$ and returns the output of $M$ on the tape.

\begin{definition}[Universal Turing Machine]
  Let $\langle M,w \rangle$ be the encoding of a \TM $M$ and an input $w$.
  A \emph{Universal \TM (\UTM)} receives input $\langle M,w \rangle$ on the tape and 
    simulates $M$ on $w$, 
  e.g., for decision problems: if $M$ accepts, then the \UTM accepts; if $M$ rejects, then the \UTM rejects. 
\end{definition}

%
In the following, we define a $\UTM$ that can simulate any other \TM on a given input for a predetermined number of execution steps, i.e., time-bounded. We focus on the following definition of a \emph{Universal Turing Machine with Time Bound} for decision problems. 
\begin{definition}[Universal Turing Machine with Time Bound~\cite{AB09}]\label{def:unitm-bounded}
    Given a \TM $M$, an input $w$, and $T \in \IN$, where the encoding of $M$, $w$ and $T$ 
        is denoted  by $\langle M,w,T \rangle$.
    A \emph{$\UTM$ with a Time Bound} receives input $\langle M,w,T \rangle$ 
        on the tape and simulates $M$ on $w$ for $T$ execution steps. 
    If $M$ accepts on input $w$ within $T$ steps, the $\UTM$ accepts; 
        if $M$ rejects on input $w$ within $T$ execution steps, the $\UTM$ rejects. 
    Otherwise, the $\UTM$ outputs some special failure symbol. 
\end{definition}
Efficient implementation of a Universal Turing Machine and a Universal Turing Machine with a Time-Bound, both of which we denote by \UTM, exists, i.e., if the simulated machine $M$ on input $x$ halts within $T$ steps, then $\UTM(\langle M, x \rangle)$ halts within $\BO(T\log T)$ steps~\cite[Thm.~1.9]{AB09}. 

\ifnum\arxiv=0
\paragraph*{Oblivious  Turing Machines}
\else
\paragraph{Oblivious  Turing Machines. }
\fi
A \TM is called \emph{Oblivious} iff its head movement depends only on the length of the input, not on the value; 
    it is formally defined as follows (paraphrased from~\cite{AB09}). 
\begin{definition}[Oblivious Turing Machine]\label{def:obliviousTM}
     A \TM\ $M$ is Oblivious if, for every input $x \in \{0,1\}^*$ and $i\in \IN$, the location of $M$’s tape-head at the $i$th step of execution on input $x$ is only a function of $|x|$ and $i$.
\end{definition}
The definition above implies that for every Oblivious \TM $M$, the execution time is $t(n)$ for every input of length $n$.
%
%
Every \TM with time complexity $t(n)$ can be simulated by an Oblivious \TM with time complexity $\BO(t(n)\log t(n))$~\cite[Ex.~1.6]{AB09}.

\subsection{Undecidability}
In Section~\ref{sec:computability}, we prove that the computation of the metastable closure (as defined below) of a function is undecidable. 
In what follows, we say that a set of encodings of \TM's is \emph{not trivial} if it contains some, but
not all \TM encodings, i.e., a set, is hence trivial if it is empty or includes all the \TM's encodings. 
To that end, we will make use of 
Rice's Theorem~\cite{Williams12}, which we define here. 
\begin{theorem}[Rice's Theorem]\label{thm:rice}
    Let $\mathcal{C}$ be a set of languages. 
    Consider the language $L_\mathcal{C}$ defined as follows 
    \begin{align*}
      L_\mathcal{C}\coloneqq\{\langle M \rangle | L(M)\in \mathcal{C}\}\,.
    \end{align*}
    Then either $L_\mathcal{C}$ is trivial or it is undecidable.
\end{theorem}



\subsection{Exponential Time Complexity}
Our main focus is on exponentially complex problems, which we define by the
    exponential time complexity class $\exptime{}$, as given in~\cite{Sipser97}.
\begin{definition}[Exponential Time Complexity]\label{def:exptime}
  We denote by $\exptime{}$ the class of decidable languages in exponential time on a deterministic \TM. 
  In other words, 
  \begin{align*}
    \exptime=\bigcup_{p(n)}\ttime\left(2^{p(n)}\right),
  \end{align*}
  where p(n) is a polynomial in $n$.
  \end{definition}

A language $L$  is called \emph{$\exptime$-hard} if any problem in $\exptime$ can be reduced in polynomial time to $L$. 
If $L$ is a language in $\exptime$, too, then $L$ is called \emph{$\exptime$-complete}.
The following definition is paraphrased from~\cite{AB09}.
\begin{definition}[\exptime-Completeness]
    A language $L \subseteq \{0,1\}^*$ is polynomial-time \emph{Karp reducible} to a language $L' \subseteq \{0,1\}^*$, denoted by $L \leq_p L'$, if there is a polynomial-time computable function $f: \{0, 1\}^* \rightarrow \{0,1\}^*$ such that for every $x \in \{0,1\}^*$, $x \in L$ iff $f(x) \in L'$. We say that $L'$ is \emph{\exptime-hard} if $L \leq_p L'$ for every $L \in \exptime$. We say that $L'$ is \emph{\exptime-complete} if $L'$ is \exptime-hard and $L' \in \exptime$.
\end{definition}
\ifnum\arxiv=0
\paragraph*{The Bounded Halting Problem}
\else
\paragraph{The Bounded Halting Problem. }
\fi
In the following, we define a decision problem known to be \exptime-complete. 
This decision problem is known as the \emph{bounded halting problem}. 
The task is to determine if a given \TM terminates in a bounded time, 
    where the time-bound is encoded in Binary Encoding.
\begin{definition}[Bounded Halting Problem (\expbhp)~\cite{DK14}]
    Given a \TM\ $M$, a string $x \in \IB^*$, and $T \in \IN\setminus \{0\} $, 
        encoded in the Binary Encoding, 
        determine whether $M$ terminates on input $x$ in at most $T$ execution steps.
\end{definition}

The problem \expbhp is \exptime-complete, as stated below. 
\begin{lemma}[Proposition 3.30~\cite{DK14}]\label{lem:bhp_exptime}
    The problem \expbhp is \exptime-complete.
\end{lemma}

\subsection{Input and Tape Alphabet}\label{sec:inputoutputalpha}
The input and tape alphabet considered in this paper are $\Sigma=\{0,1,\mfu\}$ and $\Gamma=\{0,1,\mfu, \blank\}$.
It is well known~\cite[Coro.~7.10]{HopcroftU79},\cite[Claim~1.5]{AB09} that for every $\TM$ there is an equivalent $\TM$ that uses only Boolean values as the tape alphabet (with constant overhead to the time complexity). Hence, our choice of alphabets does not limit the generality of the considered $\TM$s, i.e., computationally, the additional symbol $\mfu$ is no different. 
In Section~\ref{sec:closure}, we attach a particular semantic to $\mfu$ in the context of the Metastable Closure computation of a Boolean function. 
As motivated in the introduction, we see the $\mfu$ symbol as the nonlogical value introduced to the input in its acquisition.  
The definition of the input alphabet well captures this "distorted" acquisition, i.e., the input might contain $\mfu$ bits, but the end of the input is well marked with a \blank\ symbol. 
In Section~\ref{sec:kleene}, we discuss $\TM$s that are amenable for implementation. In turn, the considered transition functions are restricted so that a circuit can realize them. On the one hand, again, one can think of this machine as if all its alphabets are Boolean (foregoing the \mfu\ symbol); on the other, it is helpful to consider the symbol $\mfu$ as a distorted bit input to hardware realized \TM (with Boolean alphabet only) and to analyze the behavior of the realized machine. 






\section{The metastable closure}\label{sec:closure}

This section formally defines the \emph{Metastable Closure} of a Boolean Function. 
We use the symbol $\mfu$ to represent a metastable input.
Let $\IT=\{0,\mfu,1\}$ be the ternary set of logical values.
The metastable closure, which we define first, is defined by two functions: the resolution function and the superposition function. 
The \emph{resolution} function produces the set of all stable strings that a (possibly) unstable string can resolve to by replacing each occurrence of a $\mfu$ by $0$ and $1$.
\begin{definition}[Resolution~\cite{friedrichs18containing}]
    The \emph{resolution} function is denoted by $\res\colon\IT^n\rightarrow\mathcal{P}(\IB^n)$.\footnote{
    We denote by $\mathcal{P}(S)$ the \emph{power set} of the set $S$.}
    For input $x\in\IT^n$ it is defined by 
    \begin{align*}
        \res(x)\coloneqq\{y\in\IB^n|\forall i \in\{1,\ldots,n\}\colon x_i\neq\mfu\Rightarrow y_i=x_i\}\,.
    \end{align*}
\end{definition}
Informally, $\mfu$ represents the uncertainty between bits $0$ and $1$, e.g., when the input was acquired and written on the input tape, for some bits, this acquisition was not successful, hence leaving a \mfu\ instead of a 0 or 1. Note that $\blank$ symbols read from the tape represent untouched tape positions.

Next, we define the \emph{superposition}, which produces a ternary string from
    a set of input strings of equal length.
The superposition produces a string that has a $\mfu$ in every position
    where the inputs do not match.

\begin{definition}[Superposition~\cite{friedrichs18containing}]
  The \emph{superposition} is denoted by the operator $\starr\colon\IT^n\times\IT^n\rightarrow\IT^n$.
  For two words $x,y\in\IT^n$ the superposition is
  defined on each index $i\in\{1,\ldots,n\}$:
  \begin{align*}
      (x\starr y)_i\coloneqq
      \begin{cases}
        x_i, & \text{if } x_i=y_i\,,\\
        \mfu, & \text{otherwise}\,.
      \end{cases}
  \end{align*}
  As the superposition is associative, it can be naturally extended 
  to sets of words.
  \begin{align*}
    \starr \{a_0,\ldots,a_k\} \coloneqq a_0 \starr \ldots \starr a_k\,,
  \end{align*}
  where $k\in\IN$. 
\end{definition}

The \emph{metastable closure} of a Boolean function $f:\IB^n \rightarrow \IB^m$ on input $x$, denoted by $f_\mfu(x)$, 
    is computed as follows.
Apply $f$ to every possible resolution of $x$, then compute the superposition
    of all results.
We denote the application of a function $f$ to 
    every element of a set $A \subseteq \IB^n$ by $f(A)$, 
    i.e.,  $f(A) \coloneqq \{f(a) \mid a \in A\}$.

\begin{definition}[Metastable Closure~\cite{friedrichs18containing}]\label{def:closure}
  The \emph{metastable closure (\MC)} of a Boolean function 
    $f\colon\IB^n\rightarrow\IB^m$ is denoted by 
    $f_{\mfu}\colon\IT^n\rightarrow\IT^m$, for $n,m\in\IN$.
  It is defined by 
  \begin{align*}
    f_{\mfu}(x)\coloneqq\bigstarr f(\res(x))\:.
  \end{align*}
\end{definition}
\begin{table}
  \caption{Truth table of basic gates $\AND_\mfu$, $\OR_\mfu$, and $\NOT_\mfu$.}
  \label{tab:gates}
  \centering
  \begin{tabular}{c|ccc}
    $\AND_\mfu$ & 0 & 1 & \mfu\\ \hline
    0 & 0 & 0 & 0\\
    1 & 0 & 1 & \mfu\\
    \mfu & 0 & \mfu & \mfu
  \end{tabular}
  \qquad
  \begin{tabular}{c|ccc}
    $\OR_\mfu$ & 0 & 1 & \mfu\\ \hline
    0 & 0 & 1 & \mfu\\
    1 & 1 & 1 & 1\\
    \mfu & \mfu & 1 & \mfu
  \end{tabular}
  \qquad
  \begin{tabular}{c|c}
    $\NOT_\mfu$ &\\ \hline
    0 & 1\\
    1 & 0\\
    \mfu & \mfu
  \end{tabular}
\end{table}
\begin{remark}
    This paper considers \TM's that recognize Boolean functions $f\colon\IB^*\rightarrow\IB$. Definition~\ref{def:closure} extends naturally to this setting of varying input length, as the resolution is taken over strings of the same length at the input.  
\end{remark}
In \cref{tab:gates}, we define the \MC\ of basic logic operands 
    $\AND$, $\OR$, and $\NOT$, denoted by $\AND_\mfu, \OR_\mfu$ and $\NOT_\mfu$, respectively. The definition of the basic logic operands shows that even if an input is metastable, 
    the output can be stable; 
    $\AND_\mfu(0,\mfu)$ produces a $0$ at the output, hence masking the uncertainty at the input. We elaborate on the ternary logic obtained by using these extended gates in \ifnum\supp=0 Section~\ref{sec:kleenecomb}\else the supplementary material\fi. 

    In the rest of the paper, unless said otherwise, Boolean functions considered are with a \emph{single} output bit. 
\section{ The non computability of the Metastable Closure }\label{sec:computability}

In this section, we prove that, in general, computing the \MC is non-computable.
%

Define $\FMC$, the \TM that computes the \MC;
the computational task where the input is an encoding of a Boolean function $f:\IB^*\rightarrow \IB$ and a string $x \in \IT^*$, and the output is its \MC $f_{\mfu}(x)$. 
More formally, let $f$ be a Boolean function, and let $\langle M_f \rangle$ be the encoding of a \TM $M_f$ that recognizes it.  
Let $\FMC(\langle M_f \rangle,x)$ be the function that gets an encoding of a \TM  that recognizes a Boolean function $f$ and a string $x \in \IT^n$, and returns $f_{\mfu}(x)$.

In the following, we prove that the function $\FMC$ is non-computable. In the proof, we assume, towards a contradiction,  that it is computable and show that this enables us to decide an undecidable language. 

\notcomp*

%
\begin{proof}
Assume towards a contradiction that $\FMC$ is computable. We will show that it can be used to decide an undecidable language.  

Let ${\cal A}_k$ be the language of (encodings of) $\TM$s where each of the machines accepts all the inputs of the same length $k$ for some $k \in \IN\cup\{0\}$, that is, 
\[ 
    {\cal A}_k \coloneqq \{\langle M\rangle \mid  ~\forall x \in \{0,1\}^k : M \text{ accepts } x \}\:.
\]
Rice Theorem (Theorem~\ref{thm:rice}) implies that ${\cal A}_k$ is undecidable since we can rewrite ${\cal A}_k$ as $L_{\mathcal{C}_k}$, as follows:
\begin{align*}
    &{\cal C}_k\coloneqq \{L \subseteq \{0,1\}^* \mid \{0,1\}^k \subseteq L\}, \\
    &L_{\mathcal{C}_k}\coloneqq\{\langle M \rangle | L(M)\in \mathcal{C}_k\}\:.
\end{align*}
The claim follows since $L_{\mathcal{C}_k}$ is nontrivial. 

We now show that ${\cal A}_k$ can be decided, assuming that $\FMC$ is computable, i.e., there is a $\TM$ $H$ that computes $\FMC$. 

We construct a new $\TM$ $D$ in Algorithm~\ref{alg:tmdec} that decides ${\cal A}_k$, that is, upon receiving an encoding of a $\TM$ $M$ decides whether $M \in {\cal A}_k$ or not. The new $\TM$ calls $H$ as a subroutine to determine the output of $M$ on input $\mfu^k$,  where $\mfu^k$ is the word $\mfu\ldots\mfu$ of length $k$.
If $H(\langle M \rangle,\mfu^k)$ outputs a $\mfu$ or $0$, then $D$ rejects, otherwise $D$ accepts.

\begin{algorithm}
    \caption{\textbf{Turing Machine} $D$}
    \begin{algorithmic}[1]
        \State \textbf{Input:} $\langle M \rangle$, where $M$ is a Turing Machine
        \State Execute $H$ on input $(\langle M \rangle, \mfu^k)$
        \If{$H$ outputs $1$}
            \State \textbf{Accept}
        \Else
            \State \textbf{Reject}
        \EndIf
    \end{algorithmic}\label{alg:tmdec}
\end{algorithm}
    When executing $H$ on input $(\langle M \rangle, \mfu^k)$, we 
        essentially execute $M$ on every possible resolution of $\mfu^k$.
    All resolutions of $\mfu^k$ are the set of all binary
        strings of length $k$.
    Thus, if the output of $M$ accepts all resolutions, then the superposition equals $1$.  
    The $\TM$ $D$ accepts if a given input $\TM$ $M$ accepts every input of length $k$.

    Hence, as required, $\FMC$ is non-computable. 
\end{proof}
\section{The complexity of the Metastable Closure of exponential problems}\label{sec:exptime}


Given that, in general, the \MC is non-computable, in this section, we focus on the \MC of computable problems. 
We prove that given an exponential time problem with a single uncertain bit, deciding whether it has a unique resolution is \exptime-complete. 

We consider first a refinement of the bounded halting problem $\pexpbhp$, which decides if a certain \TM halts on a given input within $k$ steps, where $k$ is encoded in binary and is of the form $1z1$ where $z\in\{0\}^*$. We prove that this refined bounded halting problem retains the complexity of the original bounded halting problem, i.e., it is  \exptime-complete. 
%
We then consider the main computational problem of this section $\sizeMC$,  where there is a single $\mfu$-bit in the encoding of its instances, and prove that it is \exptime-complete by a reducing it to $\pexpbhp$. 

\ifnum\arxiv=0
\medskip
\noindent
\textbf{A Refined Bounded Halting Problem Formulation. }
\else
\paragraph{A Refined Bounded Halting Problem Formulation. }
\fi
First, we define the \emph{Bounded Halting Problem with a time bound encoded in the binary encoding of the form $10^*1$}  (\pexpbhp), i.e., the time-bound is in $\{ 2^i+1 \mid i \in \IN \setminus \{0\}\}$.
\begin{definition}[\pexpbhp]\label{def:pexpbhp}
    Given a \TM\ $M$, a string $x \in \IB^*$, and $k \in \{ 2^i+1 \mid i  \in \IN\setminus \{0\}\}$, encoded in the Binary Encoding, determine whether $M(x)$ halts in $k$ execution steps.
\end{definition}

We now prove that \pexpbhp\ is also \exptime-complete. 
Essentially, the proof works by constructing for every \expbhp\ problem with bound $k$, a \pexpbhp problem with bound $k'\geq k$ such that $k'\in\{2^i+1 \mid k \leq 2^i+1, i \in \IN\setminus \{0\}\}$.
\begin{lemma}\label{lem:powerplusone}
    The problem \pexpbhp\ is \exptime-complete.
\end{lemma}
\begin{proof}
    Let $L(\expbhp)$ and $L(\pexpbhp)$ denote the languages that correspond to the \expbhp\ and \pexpbhp\ problems, respectively. In order to prove that \pexpbhp\ is \exptime-hard it is sufficient to  prove that  $L(\expbhp) \leq_p L(\pexpbhp)$, that is, 
    we need to prove that there is a polynomial-time computable function $f: \{0, 1\}^* \rightarrow \{0,1\}^*$ such that for every $x \in \{0,1\}^*$, $x \in L(\expbhp)$ iff $f(x) \in L(\pexpbhp)$.
    
    The mapping $f$ is as follows. Given a \TM\ $M$, a string $x \in 
    \IB^*$, and $k \in \IN\setminus \{0\} $, encoded in the Binary Encoding, the mapping $f$ outputs the \TM\ $M'$, $x \in \IB^*$, and $k' = \min\{2^i+1 \mid k \leq 2^i+1, i \in \IN\setminus \{0\}\}$, where $M'$ is defined as follows:
    
    \begin{algorithm}
    \caption{\textbf{Turing Machine} $M'$}
    \begin{algorithmic}[1]
        \State \textbf{Input:} $x \in \{0,1\}^*$
        \State Execute $M$ on input $x$ for $k$ steps
        \If{$M$ terminates}
            \State \textbf{Return} $M(x)$
        \Else
            \State Enter a new sink state where the transition function loops forever
        \EndIf
    \end{algorithmic}
\end{algorithm}

    Clearly, the above mapping is polynomial.
    In case $\langle M,x,k \rangle \in L(\expbhp)$, $M(x)$ halts within $k$ steps, it halts in $k' \geq k$ steps.
    Otherwise, in case $\langle M,x,k \rangle \not\in L(\expbhp)$, $M(x)$ does not halt within $k$ steps and hence enters the new state that loops forever. In turn, $M'(x)$ does not enter the halting state within $k'$ execution steps, as required. 
    The problem \pexpbhp\ is in  $\exptime$ since \expbhp\ is an $\exptime$ problem. This shows that \pexpbhp\ is \exptime-hard.
    
    The \exptime-completeness follows from noting that every $\langle M,x,k \rangle \in L(\pexpbhp)$ is a sub-language  $L(\expbhp)$. 
    As $L(\expbhp)$ is \exptime-complete (c.f.\ \ref{lem:bhp_exptime}), the claim follows.
\end{proof}

\ifnum\arxiv=0
\medskip
\noindent
\textbf{The language \sizeMC. }
\else
\paragraph{The language \sizeMC. }
\fi
We now define a computational problem which we refer to by $\sizeMC$. A YES instance, i.e., $\langle M,x,k \rangle \in L(\sizeMC)$, if (and only if) $x$ contains {\em a single metastable bit} and there are two resolutions of the instance that yield different simulated outputs, e.g., ``not halting'' is different than $1$ (see Definition~\ref{def:unitm-bounded}).   

We prove that $L(\sizeMC)$ is $\exptime$-complete. Compare this with the trivial language that is obtained by removing the above-mentioned single uncertain bit, i.e., when all bits are stable, there is a single resolution, hence a \TM\ that decides the language always rejects.   
Moreover, this $\exptime$ completeness implies that a \TM\ that computes the MC of the corresponding decision Boolean function must run for an exponential time. 

We now formally define the $\sizeMC$ problem and its corresponding language.
\begin{definition}[\sizeMC]\label{def:expdetect}
    Let $M$ be a \TM defined on $\IB^*$, $x \in \IT^*$ is a ternary string, 
        where $x$ has exactly one $\mfu$ bit, and $k \in \IN\setminus\{0\}$ is a natural number encoded in the Binary representation. 
    Let $x_1, x_2$ be the resolutions of $x$, such that $\res(x)=\{x_1,x_2\}$.
    The task is to determine whether, after $k$ execution steps of $M$ on inputs $x_1$ and $x_2$, $M$ yields different outputs, i.e., two executions terminate and output different outputs, or one execution does not terminate.
    The language associated with this decision problem is:
    \begin{align*}
        L(&\sizeMC) \coloneqq\{\langle M, x, k \rangle |  M\text{ is a }\TM \text{ defined on }\IB^*, \\
        & x \in \IT^* , |\{x_i=\mfu \mid i \in \IN\}|=1 , k \in \IN\setminus\{0\},\\ 
        & \exists x_1,x_2\in\res(x): \UTM(\langle M,x_1,k\rangle)\neq \UTM(\langle M,x_2,k \rangle) \\
        \}\,.
    \end{align*}
\end{definition}
Note that the bound on the running time $k$ is deterministic, because no $\mfu$ bits are introduced in this part of the instance. 
The uncertainty only affects the input $x$.
Furthermore, note that the executions of $M(x_1)$ and $M(x_2)$ can terminate with the same output, a different output, only one terminates with an output, or neither terminates. 
The language $L(\sizeMC)$ only accepts instances where only one (either $M(x_1)$ or $M(x_2)$) terminates, or $M(x_1)$ and $M(x_2)$ terminate with different output, i.e., when the simulations $\UTM(\langle M,x_1,k\rangle)$ and $\UTM(\langle M,x_2,k \rangle)$ have a different outcome.
Informally, the language captures the instances, where uncertainty in the input causes uncertainty in the output, i.e., detecting a $\mfu$ bit in the output.
Similarly to $L(\sizeMC)$ we can define $L(\sizeMCm)$, where $m$ denotes the number of uncertain bits.

\ifnum\arxiv=0
\medskip
\noindent
\textbf{The problem \sizeMC is \exptime-complete. }
\else
\paragraph{The problem \sizeMC is \exptime-complete. }
\fi
\sloppy
Here we prove that the language $L(\sizeMC)$  is \exptime-complete, by reducing $\pexpbhp$ problem to the language $L(\sizeMC)$.
\exphard*
\begin{proof}
    In order to prove that $L(\sizeMC)$  is \exptime-hard it is sufficient to  prove that  $L(\pexpbhp) \leq_p L(\sizeMC)$, that is, 
    we need to prove that there is a polynomial-time computable function $f:\IB^* \rightarrow \IT^*$ such that for every $x \in \{0,1\}^*$, $x \in L(\pexpbhp)$ iff $f(x) \in L(\sizeMC)$.
    %
    %
    Recall from Definition~\ref{def:pexpbhp} that $x\in L(\pexpbhp)$ is composed of an encoding of \TM $M$, an input string $s$, and an encoding of the step counter $k$. 
    The mapping $f$ works as follows:
    (1)~define the execution step counter $\ell$, such that $\ell \leftarrow k+1$, 
    (2)~define $k'$, by replacing the leftmost bit in (the encoding of) $k$ by a $\mfu$,
    (3)~define a \TM $M'_{M,s}$ in Algorithm~\ref{alg:tmexp} that invokes $M$ on the input string $s$ for $\ell-1=k$ time steps.
    Through the mapping, $k'$ becomes the input of $M'_{M,s}$, and $\ell$ is the time bound, such that $f(\langle M,s,k \rangle)=\langle M'_{M,s},k',\ell \rangle$.
    
    Informally speaking, the idea is that $M'_{M,s}$ outputs, within $\ell$ steps,  two different outputs for the two resolutions of $k'$ if $M(s)$ halts in $k$ steps.  
    If $M(s)$ does not halt in $k$ steps, then $M'_{M,s}$ outputs a stable bit. In our case a $0$. In turn, the mapped instance is not in $L(\sizeMC)$ in this case.

    We now elaborate on parts (2) and (3) of the mapping function $f$. 
    Part~(2): Let $y$ be the binary encoding of $k$, such that for each instance of $x\in L(\pexpbhp)$, 
        $y$ has the form $1z1$, where $z\in \{0\}^*$.
    We copy each bit of $y$ to $k'$, except for the first one, replacing it by a $\mfu$.
    Hence, $k'$ has the form $\mfu0^*1$, such that $\res(k')$ only includes two binary strings of the form $00^*1$ and $10^*1$.

    Part~(3): We now elaborate on the mapping of $M$ to $M'_{M,s}$.
    The mapping is defined by a \TM in Algorithm~\ref{alg:tmexp}, which first simulates $M$ on input $s$ for $\ell-1$ steps by applying a \UTM first.
    If the \UTM stops in $\ell-1$ steps, then $M'_{M,x}$ is supposed to output a different output on every resolution of $k'$. 
    Otherwise, if the \UTM does not stop, then $M'_{M,x}$ should have an arbitrary, but consistent, output in $\IB$ for each resolution of the input $k'$. Here we choose $0$.
    Formally, the mapping to $M'_{M,x}$ is given by the following \TM:

\begin{algorithm}
    \caption{\textbf{Turing Machine} $M'_{M,s}$}
    \begin{algorithmic}[1]
        \State \textbf{Input:} $y, \langle \ell \rangle \in \{0,1\}^*$
        \State \textbf{Execute} $\UTM(\langle M, s, \ell - 1 \rangle)$ \Comment{Recall that $k=\ell-1$.}
        \If{$M$ on input $s$ halts within $\ell - 1$ steps}
            \If{$y \in 00^*1$}
                \State $o \gets 1$\label{step:5}
            \Else
                \State $o \gets 0$\label{step:7}
            \EndIf
        \Else
            \State $o \gets 0$ 
        \EndIf
        \State \textbf{Return} $o$
    \end{algorithmic} \label{alg:tmexp}
\end{algorithm}

    We now show that the above mapping satisfies for every $x \in \{0,1\}^*$, $x \in L(\pexpbhp)$ iff $f(x) \in L(\sizeMC)$.
    In case $x=\langle M,s,k\rangle\in L(\pexpbhp)$, the \TM $M$ halts on input $s$ within $k$ steps. 
    Thus, the simulation on Line~1 at $M'_{M,s}$, $\UTM(\langle M,s,\ell-1 \rangle)$, reports that $M$ stops in $\ell-1=k$ steps. 
    Note that the encoding of $k'$ has two resolutions, one resolution that encodes $k$ and one resolution that encodes $1$. 
    Hence, in Steps~\ref{step:5}-\ref{step:7} of Algorithm~\ref{alg:tmexp}, the variable $o$ is assigned a `$1$' or a `$0$' depending on the resolution, i.e., there exists $k_1,k_2\in\res(k'): \UTM(\langle M'_{M,s},k_1,\ell\rangle)\neq \UTM(\langle M'_{M,s},k_2,\ell\rangle)$.
    Hence, $f(x)=\langle M'_{M,s},k',\ell \rangle \in L(\sizeMC)$.
    Otherwise, when $x=\langle M,s,k\rangle\notin L(\pexpbhp)$, the simulation $\UTM(\langle M,s,\ell-1 \rangle)$ will stop after $\ell-1=k$ execution steps of $M$ and report that $M$ on input $s$ did not stop.  
    Hence, the \TM $M'_{M,s}$ will output $0$ for each resolution of $k'$, such that for all $k_1,k_2\in\res(k'): \UTM(\langle M'_{M,s},k_1,\ell\rangle) = \UTM(\langle M'_{M,s},k_2, \ell \rangle) = 0$.
    Hence, $f(x)=\langle M'_{M,s},k',\ell\rangle \notin L(\sizeMC)$.
    
    The above mapping is polynomial.

    The problem \sizeMC\ is in  $\exptime$ since given an instance $\langle M, x, k \rangle$ one can simulate the two executions of $\UTM(\langle M, x_1, k \rangle)$ and $\UTM(\langle M, x_2, k \rangle)$ with a running time of $\BO(k\log k)$ (see the text following  Definition~\ref{def:unitm-bounded}) which is at most exponential in the encoding length of $k$. 
\end{proof}

Theorem~\ref{thm:exphard} proves that any \TM that decides $L(\sizeMC)$ has a running time that is exponential in its input length even though it contains a single uncertain bit. Observe that a \TM that computes the metastbale-closure of this decider also decides $L(\sizeMC)$: if for a given instance, the two evaluations that correspond to the two resolutions at hand yield two different values, then the \MC of this computation yields a $\mfu$ (or an encoding of it) in the output. Otherwise, if the evaluations result in the same values (e.g., both do not terminate), then the output does not contain a $\mfu$. In turn, an instance is a YES instance if and only if the decider outputs a $\mfu$. Moreover, M-containing implementation of a universal Turing machine is also a decider of $\sizeMC$ as one applies it on $\langle M, x, k \rangle$ and accepts or rejects according to the universal Turing machine's output.

\begin{remark}
        %
        We show in Section~\ref{sec:simulation}, that there is an \MC universal \TM that performs an exponential number of steps for instances that halt, e.g., bounded execution of $k$ steps of the respective machine.  
        Moreover, our suggested \MC universal \TM is oblivious and has a (strong) restriction on the transition functions allowed. This restriction enables a rather easy realization in hardware of this machine, i.e., by using combinational circuits and memory modules (see Sections~\ref{sec:kleene},\ref{sec:simulation}). 
        %
\end{remark}
\section{Computing the size of the \IRS\ of a polynomial problem in \PP\ is \CONP-complete}
\label{sec:coNP}
In the previous section, we proved that computing the size of the \IRS\ of an exponential time problem with a single uncertain bit is \exptime-complete. In contrast, the situation is very different for a \emph{polynomial time}  problem, i.e., a problem in \PP. Here, resolving a small set (of logarithmic size)  of uncertain bits is in \PP, while resolving \emph{any number} of uncertain input bits, as we prove in this section, is \CONP-complete.

For the first fact, note that if we have an $O(\log n)$ uncertain bits, out of $n$ input bits, we can simulate all $2^{O(\log n)}$ possibilities, which is polynomial in $n$. Computing each possible resolution is polynomial, so overall we end up with a polynomial number of computations to compute the \IRS. 

The \CONP-complete language \Tautology\ (e.g.,~\cite [Ex.~2.21]{AB09}) is the language of all (encoding of) formulas that evaluate to $1$ for all Boolean inputs.~\footnote{A \emph{formula} is a combinational circuit where each gate's output feeds a single input port, i.e., every gate has a fan-out of $1$. } 

%
\ifnum\arxiv=0
\textbf{The language \sizearbMC. }
\else
\paragraph{The language \sizearbMC. }
\fi
Analogously to \sizeMC, we now define a computational problem denoted as $\sizearbMC$. On the one hand, we eliminate the assumption regarding the number of $\mfu$ bits in the ternary input $x$. On the other hand, we introduce the constraint that the step counter must be polynomial in the length of $x$.

\begin{definition}[\sizearbMC]\label{def:exparbdetect}
    Let $M$ be a \TM defined on $\IB^*$, $x \in \IT^*$ is a ternary string, where $c \in \IN\setminus\{0\}$ is a constant encoded in the binary representation.
    The task is to determine whether there are inputs $x_1, x_2 \in \res(x)$ s.t. after $|x|^c$ execution steps of $M$ on $x_1$ and on $x_2$, $M$ yields different outputs, i.e., two executions terminate and output different outputs, or one execution does not terminate.
    The language associated with this decision problem is:
    \begin{align*}
        L(&\sizearbMC) \coloneqq\{\langle M, x, c \rangle |\\  & M\text{ is a }\TM \text{ defined on }\IB^*, \\
        & x \in \IT^* , c \in \IN\setminus\{0\}, \\ 
        & \exists b\in \IB ~\forall x' \in\res(x): 
        \UTM (\langle M,x',|x|^c\rangle) = b
        \}\,.
    \end{align*}
\end{definition}

\medskip
\noindent

%
\ifnum\arxiv=0
\textbf{The problem \sizearbMC is \CONP-complete. }
\else
\paragraph{The problem \sizearbMC is \CONP-complete. }
\fi
Here we prove that the language $L(\sizearbMC)$  is \CONP-complete, by reducing the $\Tautology$ language to the language $L(\sizearbMC)$.

\coNP*
\begin{proof}
    To prove that $L(\sizearbMC)$  is \CONP-hard it is sufficient to prove that  $\Tautology \leq_p L(\sizearbMC)$, that is, we need to prove that there is a polynomial-time computable function $f:\IB^* \rightarrow \IT^*$ such that for every formula $\tau$, $\langle\tau\rangle \in \Tautology$ iff $f(\langle\tau\rangle) \in L(\sizearbMC)$.
    %
    %

    The mapping $f$ works as follows:
    (1)~\emph{Mapping to $M$ and $c$:} $M$ is the  \TM $M$ defined in Algorithm~\ref{alg:TMcoNP} that gets as an input an encoding of a formula $\tau$ and a corresponding Boolean input $y$. The \TM $M$ invokes another \TM which we refer to as $M_{\eval}$ that evaluates the input formula $\tau$ on the input $y$. The formula evaluation procedure is polynomial in the formula's size (see e.g.,~\cite[Sec.~6.3]{digirig2012}); hence, there is a constant $c$ such that the evaluation of the formula $\tau$ on $y$ ends within $(|\tau|+|y|)^c$ steps. Furthermore,  $M$ completes its computation within $(|\tau|+|y|)^{c'}$ steps for a constant $c' \geq c$, and 
    (2)~\emph{Mapping to $x$:} Let $n$ denote the number of inputs to the formula $\tau$, then define $x$ to be the encoding of $\tau$ concatenated to $\mfu^n$. Note that in all resolutions of $x$ the encoding of $\tau$ does not change. In turn, $f(\langle\tau\rangle) = \langle M, \tau \circ \mfu^n, c'  \rangle$. 
    

    We now elaborate on part (1) of the mapping function $f$. 
    The mapping is defined by a \TM, which first simulates the formula $\tau$ on the all $0$'s input (see Step~\ref{step:contradiction} in Algorithm~\ref{alg:TMcoNP}). This step (and steps that follow it) comes to enforce $2$ different outputs in the case where $\tau$ is a \emph{contradiction} (i.e., a formula that evaluates to $0$ on all inputs). In case of such a formula, Step~\ref{step:contradiction} makes sure that there are two inputs of length $n$ that result in two different outcomes. Otherwise, $M$ returns $\tau$'s evaluation on the input $y$ by simulating the $M_\eval$ \TM on the input $\langle \tau, y\rangle$. In turn, the evaluation always halts within the predetermined polynomial time and outputs the same output as the evaluation of the formula $\tau$. 

\begin{algorithm}
\caption{\textbf{Turing Machine} $M$}
    \begin{algorithmic}[1]
        \State \textbf{Input:} a formula $\tau$ with $n$ inputs and $y \in \{0,1\}^n$
        \State \textbf{Execute} $\UTM(\langle M_{\eval}, \tau, 0^n \rangle)$
        \If{$M_{\eval}$ on input $\langle \tau, 0^n\rangle$ returns $0$} \Comment{Making sure that contradictions do not yield $|\tau(\res(x)|=1$. }\label{step:contradiction}
            \If{$y = 0^n$}
                \State $o \gets 0$
            \Else
                \State $o \gets 1$
            \EndIf
        \Else
            \State $o \gets \UTM(\langle M_{\eval}, \tau, y  \rangle)$\label{step:taut}
        \EndIf
        \State \textbf{Return} $o$ 
    \end{algorithmic}
\label{alg:TMcoNP}  \end{algorithm}
    
    Consequently, if $\langle\tau\rangle \in \Tautology$, then by Steps~\ref{step:contradiction},~\ref{step:taut}, and since $|y| = |x|$  for any $y \in \res(x)$, it follows that $\forall y \in \res(x)$ it holds that $\UTM (\langle M,y,(|\tau|+|x|)^{c'}\rangle) = 1$. Furthermore,  if $\langle\tau\rangle \not\in \Tautology$, then by Step~\ref{step:contradiction}  $\exists y_1,y_2 \in \res(x)$ s.t.  $\UTM (\langle M,y_1,(|\tau|+|x|)^{c'}\rangle) \neq  \UTM (\langle M,y_2,(|\tau|+|x|)^{c'}\rangle)$.
    
    
    The above mapping is polynomial.

    The problem \sizearbMC\ is in  $\CONP$ since its complement problem is in \NP, that is, given an instance $\langle M, x, c \rangle$ and the witness $x_0 \in \IT^*$ and $x_1,x_2 \in \res(x_0)$ one can simulate the two executions of $\UTM(\langle M, x_1, |x_0|^c \rangle)$ and $\UTM(\langle M, x_2, |x_0|^c \rangle)$ with a polynomial running time and verify that indeed these two executions result with different outcomes.~\footnote{Note that $|x|^c$ is encoded in the binary encoding by $\BO(c\log|x|)$ bits. The overall simulation is polynomial in $|x|^c$, which is polynomial in $|x|$.} 
\end{proof}

\section{Representing metastable Turing Machines using the Kleene Logic}\label{sec:kleene}\label{sec:tmkleene}
This section deals with Turing Machines, whose transition function is restricted to \emph{Natural functions}, i.e., functions that can be implemented by basic gates extended to their \MC versions. These functions generalize standard Boolean functions to the ternary Kleene logic, such that for a Natural function $f :\IT^n \rightarrow \IT$ it holds that $f(\{0,1\}^n) \subseteq \{0,1\}$ and that uncertainty is monotone, i.e., increasing the number of uncertain input bits does not decrease the output uncertainty 
\ifnum \supp=0
(see Thm~\ref{thm:natural} in Appendix.~\ref{sec:kleenecomb} for a formal characterization of this set of Natural functions; Appendix~\ref{sec:kleenecomb} contains some basic material on representing Boolean functions in the Kleene logic).  
\else
(see the supplementary material for the full circuit-theoretic background).
\fi
The fact that the transition function is Natural limits the computed ternary functions per transition, 
as shown in this section. However, the discussion includes all standard (Boolean) Turing machines with meta-stable inputs as a special case.

Finally, a \TM whose transition function is Natural is implementable by a combinational circuit and (simple) memory registers - we refer to this subclass of machines by \emph{Natural Turing Machines}. 

%
%
%
We now revisit the definition of Turing Machines\ref{def:dettm} and define a subclass of Turing Machines that are implementable by Natural Boolean functions. We refer to such \TM's by \emph{Natural Turing Machines.}
For the description of this subclass of \TM's we take a less abstract approach and describe these machines by their functionality, i.e., how the transition function and output function behave, how states are encoded and how the reading head is managed. 
The input tape might include $\mfu$ symbols that, when read, affect the output of the circuit that implements the transition and output functions. This is precisely one of the points where the more circuit-oriented description of Natural \TM's comes to play: how combinational circuits implement these functions dictates the next state and the output.

\emph{Why Natural \TM's are a subclass of all \TM's?} 
The discussion in the beginning of this section (as formalized in \ifnum\supp=0 Theorem~\ref{thm:natural}\else  the supplementary material\fi) implies that not all ternary functions are implementable by combinational circuits. As an example, let us consider two such (desirable) operations. In both operations, a single bit is given as an input and a single bit is output. In the first function $r:\IT \rightarrow \IT$ the input bit is output as is if it is a stable one, or it is resolved to, say, the value of $1$. A somewhat related function $d:\IT \rightarrow \IT$ detects whether this bit is stable. \ifnum\supp=0\cref{thm:natural} implies that \else In turn, \fi these two operations are not natural since they violate the monotonicity w.r.t.\ $\preceq$ property and are not implementable by combinational circuits, as formalized in the following corollary. 

\begin{corollary}[[Coro.~17,18~\cite{friedrichs18containing}]
    The functions $d,r\colon\IT\rightarrow\IB$, with
    \begin{align*}
        d(x)\coloneqq
        \begin{cases}
            0 & \text{if } x\in\IB \\
            1 & \text{if } x=\mfu
        \end{cases}
        \quad
        r(x)\coloneqq
        \begin{cases}
            x & \text{if } x\in\IB \\
            1 & \text{if } x=\mfu
        \end{cases}
    \end{align*}
    are not monotone, they
    cannot be implemented by a combinational circuit.
\end{corollary}
Hence, not all transition and output functions are possible when discussing input and output alphabets of $\IT$.
Last detail of the description of Natural Machines is how the movement of the reading head is defined. Since combinational circuits implement the output function, they might output $\mfu$s  to be written on the machine's tape and output a $\mfu$ to the direction to which the reading head should move next. We hence consider \emph{Oblivious} machines here. This eliminates this source of ambiguity of the direction of the moving head.

In the next section we show a \emph{Universal Natural \TM} implementation that computes the \MC\ of a given encoding of a \TM\ and its input. 

We are now ready to describe the class of Natural \TM's. In this definition, and from an abstract point of view, all the symbols are encoded in binary (this requires three bits to encode $\{0,1,\blank,\mfu,\leftt,\rightt\}$). As mentioned in Section~\ref{sec:inputoutputalpha} such encodings do not restrict computability. 

\begin{definition}[A Natural Turing Machine]\label{def:nattm}
  A \emph{Natural Turing Machine} is a $7$-tuple, $(Q, \Sigma, \Gamma,\delta, q_0, \blank, F)$,
    where $Q$, $\Sigma$, $\Gamma$ are all finite sets and
  \begin{itemize}
      \item $Q = \{0,\mfu,1\}^q$ is the set of states,
      \item $\Gamma=\{0,\mfu,1,\blank\}$ is the tape alphabet,
      \item $\blank$ is the blank symbol,
      \item $\Sigma=\{0,\mfu,1\}$ is the input alphabet, 
      \item $\delta\colon Q\times\Gamma\rightarrow Q\times\Gamma\times\{\leftt,\rightt\}$ is the transition function,
      \item $q_0\in \{0,1\}^q$ is the initial state,
      \item $F\subseteq \{0,1\}^q$ is the set of final states.~\footnote{This definition of final states does not limit the generality of the model. One can connect all the final states to a ``super'' final state which is stable. In turn, if there is an execution that ends up in the closure of final states, then in the next step all of converge to the super state. } 
  \end{itemize}
  Moreover, Natural Turing Machines are Oblivious and the Boolean transition function $\delta$ is a Natural Function.  
\end{definition}

The behaviour of Natural \TM's is exactly as described in Section~\ref{sec:tm}. Again, the restrictions posed here are that the transition function is a Natural Function (e.g., implemented by a combinational circuit), and the initial and final states must be ``stable''.  

From a hardware point of view (foregoing the technicality of the blank symbols) the combinational circuit $C_\delta$ implements the $\delta$ function restricted to Binary inputs. 
\ifnum\supp=0
Theorem~\ref{thm:natural} implies that the way this circuit is designed yields its behavior facing $\mfu$'s. 
\fi
The initialization to the registers holding the states of the Natural \TM\ are initialized to a string in $\{0,1\}^q$ that encodes the initial state $q_0$. Any $\mfu$ input  that enters $C_\delta$ might affect its outputs which are either output to the tape or saved in the registers holding the state, which in this case captures the ambiguity imposed by the input $\mfu$ - this is the reason that the state space in Definition~\ref{def:nattm} is $\{0,\mfu,1\}^q$.~\footnote{Here we use \emph{simple registers} as defined by Friedrichs et al.~\cite{friedrichs18containing} which simply capture the signals in their inputs output it as is. }
%


%

\section{Implementation of a Natural and Universal M-Containing Turing machine} 
\label{sec:simulation}
In this section, we present a Natural and Universal M-Containing Turing machine.
That is, given a \TM $M$ and input $x\in\IT^*$, the Natural \TM simulates $M_\mfu(x)$.
The Natural Universal M-Containing \TM is denoted by $\UMCTM$. 
Intuitively, the construction works as follows:
\begin{itemize}
    \item When the length $n$ of the input is known, enumerate all
        possible stable inputs, i.e., the set $\IB^n$.
    \item Simulate $M$ on every input $i\in\IB^n$ and take notes
        of the output on the tape.
    \item Simulate a $\CMUX$ that selects the tape entry corresponding to input $x$.
\end{itemize}
The $\CMUX$ correctly selects the right entry if $x$ is stable,
        otherwise, it returns the \MC of all entries
        corresponding to resolutions of $x$.
Before presenting the construction, we show how to 
    simulate the $\CMUX$ in the following section.

The construction draws inspiration from the exponential circuit construction
    in~\cite{IkenmeyerKLLMS19}, which computes the \MC for arbitrary 
    inputs by controlling where $\mfu$ enters the computation.
Similarly to the exponential circuit construction, $\UMCTM$ has an exponential 
    overhead in the running time. 
By the lower bound in \cref{sec:exptime}, we can conclude that the exponential 
    overhead is necessary and the construction is asymptotically optimal.

\subsection{Simulating an M-Containing Multiplexer \CMUX}\label{sec:simcmux}

Here we describe how to use a \TM to simulate an M-Containing Multiplexer \CMUX as 
    given in 
    \ifnum\supp=0
    \cref{def:cmux}.
    \else
    the following definition.
    \begin{definition}[M-Containing Multiplexer (\CMUX)]\label{def:cmux}
    Given $\ell\in\IN_{>0}$, such that $n=2^\ell$;
        an $n$ to $1$ \emph{multiplexer} is given by the Boolean function  
        $\NMUX\colon\IB^n\times\IB^\ell\rightarrow\IB$. 
    Let $x\in\IB^n$ be a string of input bits and $s\in\IB^\ell$ be the select input, then
    \begin{align*}
        \NMUX(x,s)\coloneqq x_p \text{, and } p \coloneqq {\langle s\rangle_B}\,,
    \end{align*}
    where $\langle s \rangle_B \in \IN$ denotes the number encoded (in binary encoding) by $s$,
    and $x_p$ denotes the bit at position $p$ in $x$.
    
    The $n$ to $1$ \emph{M-Containing Multiplexer ($\CMUX$)} is given by the function 
        $\NCMUX\colon\IT^n\times\IT^\ell\rightarrow\IT$.
    It is defined by the M-Containing of the $n$ to $1$ multiplexer;
    \begin{align*}
        \NCMUX(x,s) \coloneqq \left(\NMUX(x,s)\right)_\mfu\,,
    \end{align*}
    where $x\in\IT^n$ and $s\in\IT^\ell$.
\end{definition}
    More about M-Containing Multiplexers can be found in the  supplementary material.
    \fi

Inputs to the \TM are then the string $x\in\IB^n$, and select input 
    $s\in\IB^\ell$, under the assumption that
    $n=2^\ell$ is initially known to the \TM.
%
%
The initial configuration is given as follows; 
\begin{itemize}
    \item $s$ is given in binary encoding on the input tape,
    \item $x$ is placed on the internal tape,
    \item $n,\ell$ are known initially to the \TM,
    \item $x,s$ are followed by $\blank$'s. 
\end{itemize}
A \CMUX on single bits $a,b,s\in\IT$ can be implemented with a few gates in natural logic, as follows;~\cite{FriedrichsK17}
\begin{align}
    &\CMUX(a,b,s)\coloneqq\notag\\
    \OR_\mfu(&
        \AND_\mfu(a,\NOT_\mfu(s)), 
        \AND_\mfu(b,s),
        \AND_\mfu(a,b))\,.\label{eqn:cmux}
\end{align}

We denote \TM that can simulate an $\NCMUX$, by $T_{\CMUX}$.
Intuitively, the machine $T_{\CMUX}$ 
\begin{itemize}
    \item splits the result of the previous step into two halves 
    \item processes the two halves,
    \item feeds them bit-by-bit into a \CMUX,
    \item and overwrites the first positions of the tape.
\end{itemize}
The final result is, hence, stored in the first bit of the tape.
Formally, $T_{\CMUX}$ is given as follows;

\begin{algorithm}
    \caption{\textbf{Turing Machine} $T_{\text{CMUX}}$}
    \begin{algorithmic}[1]
        \State \textbf{Input:} $x \in \{0,1\}^n$, $s \in \{0,1\}^\ell$
        \For{$i = 0$ to $\ell - 1$}
            \For{$j = 0$ to $2^{\ell - i - 1} - 1$}
                \State Read $a \gets x[j]$
                \State Read $b \gets x[j + 2^{\ell - i - 1}]$
                \State Write $x[j] \gets \text{CMUX}(a, b, s[i])$
            \EndFor
        \EndFor
        \State \textbf{Return} $x[0]$
    \end{algorithmic}
\end{algorithm}

\ifnum\arxiv=0
\begin{figure}
    \centering    \includegraphics[width=0.6\linewidth]{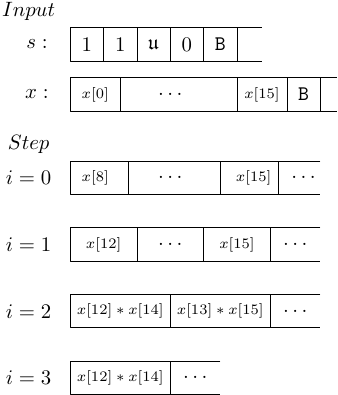}
    \caption{Example execution of $T_\CMUX$, where $n=4$, $s=11\mfu0$ and $x=x_0\ldots x_{15}$.
    Each step shows the tape content after executing the For-loop in line \textbf{1.} for one round.}
    \label{fig:expsimu_ex}\label{fig:cmux} 
\end{figure}
\else
\begin{figure}
    \centering    \includegraphics[width=0.4\linewidth]{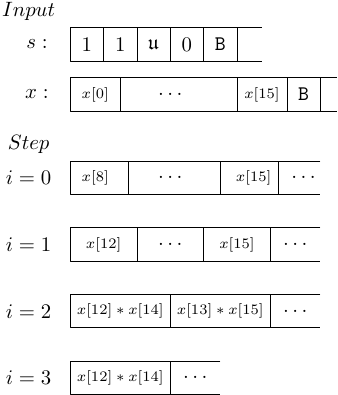}
    \caption{Example execution of $T_\CMUX$, where $n=4$, $s=11\mfu0$ and $x=x_0\ldots x_{15}$.
    Each step shows the tape content after executing the For-loop in line \textbf{1.} for one round.}
    \label{fig:expsimu_ex}\label{fig:cmux} 
\end{figure}

\paragraph{Example\ifnum\arxiv=1 .\fi}
In \cref{fig:cmux} an example for $\ell=4$ and $s=11\mfu0$ is given.
If we replace $\mfu$ by $0$ and $1$, we observe that the resolutions of $s$ correspond to the binary encodings of $12$ and $14$, i.e., $1100=\langle 12 \rangle_B$ and $1110=\langle 14 \rangle_B$.
Execution of steps $i=0$ (resp.\ $i=1$) copies $x[8]$ to $x[15]$ (resp.\ $x[12]$ to $x[15]$) to the front of the tape, because $s[0]=1$ (resp.\ $s[1]=1$).
Due to the implementation of the single bit \CMUX (Equation~\ref{eqn:cmux}), step $i=2$
computes the superpositions $x[12]*x[14]$ and $x[13]*x[15]$.
Since $s[3]=0$ is stable, the \CMUX correctly selects the tape position with $x[12]*x[14]$.  

\paragraph{The \TM simulating an $\NCMUX$ is Oblivious and Natural\ifnum\arxiv=1 .\fi}
The position of the tape head at each computational step depends only on $\ell$ and the number of the computational step.
Hence, $T_{\CMUX}$ is Oblivious according to \cref{def:obliviousTM}.
Moreover, $T_\CMUX$ is a Natural \TM, every write operation is the result of reading the content of the tape and processing it with gates in Kleene logic (see \cref{eqn:cmux}). 

\begin{lemma}\label{lem:cmuxnatural}
    The $T_\CMUX$ computes the $\NCMUX$ function on inputs $x$ and $s$ as given in \cref{def:cmux}, it is an Oblivious and Natural \TM.
    The $T_\CMUX$ has running time $\BO(n)$.
\end{lemma}

\subsection{Universal M-Containing Turing Machine}
We are now in the position to describe $\UMCTM$.
Let $M$ be a \TM with running time $t(n)$ and $x\in\IT^n$ be the input on which $M$ is simulated.
Informally, $\UMCTM$ 
\begin{itemize}
    \item enumerates all possible resolutions of $x$
    \item simulates $M$ on each of the stable resolutions
    \item use subroutine $T_\CMUX$, with input $x$ as the 
    selection string, to select between all stable inputs.
\end{itemize}
Formally $\UMCTM$ is given in the following algorithm;
\begin{algorithm}
    \caption{\textbf{Universal M-Containing Turing Machine} $\UMCTM$}
    \begin{algorithmic}[1]
        \State \textbf{Input:} $M$ is a Turing Machine, $x \in \IT^*$
        \State $n \gets |x|$
        \ForAll{$i \in \{0,1\}^n$}
            \State $y[\langle i \rangle_B] \gets M(i)$
        \EndFor
        \State Simulate $c \gets T_{\text{CMUX}}(y[0:2^n - 1], x)$
        \State \textbf{Return} $c$
    \end{algorithmic}
\end{algorithm}


Initially, $\UMCTM$ starts by counting the length of the input $x$ (see line {\bf 1.}).
Because the $\blank$ symbol is not affected by metastability, the \TM can traverse 
    the input tape until reading the first $\blank$ and count the number of steps.
No increment is affected by metastable inputs.


Once the size of the input is known, the \TM $M$ can be simulated on every
    stable input $i\in\IB^n$.
The output of $M(i)$ is stored on the tape at position $\langle i\rangle_B$ (see line {\bf 2.}).~\footnote{
The choice of the encoding does not matter as long as the encoding is bijective and the $\CMUX$ uses the same encoding.
We choose the binary encoding, because it allows for optimal construction of the \CMUX.}
We use an oblivious \TM to simulate $M$, as any \TM $M$ can be simulated by an oblivious \TM in running time $\BO(t(n)\log t(n))$~\cite{AB09}.

Last, line {\bf 3.} uses the $T_\CMUX$ subroutine to select the run of $M$ corresponding to input $x$.
Since $T_\CMUX$ is M-Containing, the subroutine correctly chooses the superposition over all runs that 
    correspond to resolutions of $x$.

\paragraph{Correctness and running time of $\UMCTM$\ifnum\arxiv=1 .\fi}
The correctness of $\UMCTM$ follows from the construction of $\UMCTM$ and the correctness of $T_\CMUX$.

\mcuni*
\begin{proof}
    Due to correctness of $T_\CMUX$ (c.f.\ \cref{lem:cmuxnatural}), 
        line \textbf{3.} computes
    \begin{align*}
        c &=\CMUX(y[0:2^n-1], x)\\
        &=\bigstarr \{\MUX(y'[0:2^n-1], x')\,|\, x' \in \res(x), y' \in \res(y) \}\\
        &=\bigstarr \{y'[\langle x' \rangle_B]\,|\,x' \in \res(x), y' \in \res(y) \}\\
        &=\bigstarr \{y[\langle x' \rangle_B]\,|\,x' \in \res(x)\}
    \end{align*}
    The last equation is given by the fact that $y$ is generated from stable bits only, such that $y\in\IB$, and $\res(y)=\{y\}$.
    By the \textbf{For}-loop in line \textbf{2.} we obtain
    \begin{align*}
        y[\langle x' \rangle_B] = M(x')\,.
    \end{align*}
    Hence, the $\UMCTM$ on inputs $M$ and $x$ computes
    \begin{align*}
        c = \bigstarr \{M(x')|x'\in\res(x)\} = M_\mfu(x)\,.
    \end{align*}
    %
    We regard the running times of lines \textbf{1.}, \textbf{2.}, and \textbf{3.} separately.
    Line \textbf{2.} dominates the running time, due to exponentially many 
        simulations of the \TM $M$, lines \textbf{1.} and \textbf{3.} 
        have only linear running time. 
    Any \TM $M$ of running time $t(n)$ can be simulated in time 
        $\BO(t(n)\log t(n))$ (see \cite{AB09}, Theorem 1.13).
    Simulation of $M$ is invoked $2^n$ many times.
    Thus, hiding the $\log$ factor, we can bound the running time of $\UMCTM$ 
        by $\tilde{\BO}(t(n)\cdot{2^n})$.
\end{proof}

As stated earlier, the \TM $\UMCTM$ is a Natural \TM, hence, it can be implemented with 
natural gates using Kleene logic. The claim follows from the construction of $\UMCTM$ and
Theorem~\ref{lem:cmuxnatural}.

\begin{lemma}
    The M-Containing \TM $\UMCTM$ is a \emph{Natural} and \emph{Oblivious} \TM.
\end{lemma}
\begin{proof}
    We argue for each line separately that it can be implemented by a Natural function.
    Line \textbf{1.} is fully stable, and no $\mfu$ enters the system, since it traverses the input
    (without reading it) until the first $\blank$ symbol.
    The simulation of $M$ in line \textbf{2.} is performed on values in $\IB$ only.
    After execution of the \textbf{For}-loop in line \textbf{2.} the tape contains only values in $\IB$. 
    Hence, the implementation of lines \textbf{1.} and \textbf{2.} can be done with natural gates.
    
    In line \textbf{3.} the output is computed by a subroutine simulating a \CMUX on input $x$ and the tape content.
    The claim follows from \cref{lem:cmuxnatural}, which states that the \CMUX subroutine is natural.

    Similarly, we argue for each line that the computation is Oblivious.
    Since line \textbf{1.} traverses the input once from left to right, it is oblivious. 
    The number of simulations in line \textbf{2.} depends only on the length of the input. 
    Since we are using an Oblivious \TM to simulate $M$, hence line \textbf{2.} is oblivious, too.
    From Lemma~\ref{lem:cmuxnatural} it follows that the \CMUX subroutine is oblivious.
    This concludes the claim.
\end{proof}

\begin{remark}
    \begin{itemize}
        \item The $\UMCTM$ does not necessarily have to be fully Oblivious. We can use a non-Oblivious \TM dor the  simulation in Line \textbf{2.}, which might save simulation time if $M$ is non-Oblivious. However, asymptotically the running time of the $\UMCTM$ is identical for Oblivious and non-Oblivious computation.
        \item If \TM $M$ has time complexity in \exptime or harder, then simulating $M$ with $\UMCTM$ remains in the complexity class \exptime.
        Asymptotically, the construction does not add any blowup.
    \end{itemize}
\end{remark}

\begin{corollary}\label{coro:unioblopt}
    When using an Oblivious simulation in line \textbf{3.} of $\UMCTM$, 
    then $\UMCTM$ is Oblivious.
\end{corollary}

\medskip
\textbf{Generality of the Universal M-Containing Turing Machine.} 
At first glance, the $\UMCTM$ seems to be restricted to inputs $\langle M,x\rangle$, 
where $M(x)$ always terminates.
However, we can make use of the \UTM with a Time Bound (c.f.\ Definition~\ref{def:unitm-bounded}) to compute any \TM for a given number of steps.
Given an arbitrary \TM $M$, an input $x$, and a time bound $T$, the \UTM simulates $M(x)$ for $T$ time steps. 
%
We can make use of the \UTM with a Time Bound to compute the \MC of $M$ as follows;
\begin{align*}
    \UMCTM(\UTM,\langle M,T,x\rangle)\,,
\end{align*}
the \TM given to $\UMCTM$ is \UTM, and the argument is $\langle M,x,T \rangle$.
Such that $\UMCTM$ (1) lists all possible inputs of length $|\langle M,x,T \rangle|$, (2) executes $\UTM$ on each possible input, and (3) selects the output with the \CMUX subroutine. 
Because $M$ and $T$ do not contain a $\mfu$, the \CMUX will select the correct output corresponding to the \MC of $M$ on input $x$ after $T$ time steps.

Note that there might be a resolution $y\in\res(x)$, such that $M(y)$ does not terminate within $T$ steps.
In this case, $\UTM(M,y,T)$ will output a special symbol which represents the case that $M(y)$ does not terminate within the time bound.
We extend the superposition to the special symbol in the intuitive sense, i.e., the superposition of a set containing only the special symbol returns the special symbol, otherwise, the superposition of a set containing the special symbol and other ternary symbols returns a $\mfu$.
Hence, if for all resolutions $y\in\res(x)$, $M(y)$ does not terminate, then the $\UMCTM$ returns the special symbol.
Otherwise, if there exists one input resolution that terminates in $T$ steps and there exists one resolution that does not terminate, then the $\UMCTM$ returns $\mfu$, i.e., the superposition of the output of $M(y_1)$ and the special symbol.

We make the statement more general by allowing uncertainty in the encoding of $M$ and $T$ as follows;
The machine simulated by the $\UMCTM$ is the $\UTM$ and the input is $\langle M,x,T\rangle\in\IT^*$.
Before simulating it, the $\UTM$ needs to verify that $M$, for some resolution of $\langle M,x,T\rangle$, is a representation of a \TM.
Similar to~\cite{AB09}, we propose to first parse the encoding of $M$.
Invalid strings are mapped to a default \TM, which immediately terminates and outputs $0$.
The $\UTM$ first needs to read the set of states, the input alphabet, and the output alphabet. Then, the the $\UTM$ needs to ensure that the transition function is defined only on valid states and symbols.
The parsing time is polynomial in the length of the encoding of $M$.

An efficient implementation of a \UTM with a Time Bound can simulate its input in $\BO(T \log T)$ steps (see text following Definition~\ref{def:unitm-bounded}).
Let $n_M = |\langle M \rangle|$ be the length of the encoding of $M$, then simulation with a \UTM including an input parser requires running time $\BO(n_M + T\log T)$.
Assuming that the Bound on the running time is much larger than the size of the encoding, this term is dominated by $T\log T$.
Plugging everything into Theorem~\ref{thm:mcuni}, we obtain a running time of $\tilde{\BO}(T\cdot 2^{n'})$, where $n'=|\langle M,x,T \rangle|$.
The result in Section~\ref{sec:exptime} suggests that this running time is optimal.

\begin{corollary}\label{cor:generality}
    Let \TM $M$ be defined on $\IB^*$, input $x\in\IT^*$, and time bound $T\in\IN\setminus\{0\}$.
    Given a possibly uncertain encoding $\langle M,x,T\rangle\in\IT^*$, 
    we can compute the \MC over all possible resolutions of $M$, $x$, and $T$ by
    \begin{align*}
        \UMCTM(\UTM,\langle M,x,T\rangle)\,.
    \end{align*}
    in $\tilde{\BO}(T\cdot 2^{n})$ running time, where $n$ is the length of the encoding of $M$, $x$, and $T$.
\end{corollary}

\section*{Acknowledgments}
During the preparation of this manuscript, the authors used AI-based tools exclusively for editorial support, limited to grammar, language polishing, and readability improvements.

%
\ifnum\arxiv=0
\bibliographystyle{IEEEtran}
\else
\bibliographystyle{alpha}
\fi
\bibliography{biblio.bib}

@book{digirig2012,
  title={Digital logic design: a rigorous approach},
  author={Even, Guy and Medina, Moti},
  year={2012},
  publisher={Cambridge University Press}
}

@article{turing1936computable,
  author    = {Turing, Alan M.},
  title     = {On Computable Numbers, with an Application to the Entscheidungsproblem},
  journal   = {Proceedings of the London Mathematical Society},
  volume    = {s2-42},
  number    = {1},
  pages     = {230--265},
  year      = {1936},
  doi       = {10.1112/plms/s2-42.1.230}
}

@article{church1936unsolvable,
  author    = {Church, Alonzo},
  title     = {An Unsolvable Problem of Elementary Number Theory},
  journal   = {American Journal of Mathematics},
  volume    = {58},
  number    = {2},
  pages     = {345--363},
  year      = {1936},
  publisher = {Johns Hopkins University Press},
  doi       = {10.2307/2371045}
}

@article{BN20,
  title={On the complexity of detecting hazards},
  author={Komarath, Balagopal and Saurabh, Nitin},
  journal={Information Processing Letters},
  volume={162},
  pages={105980},
  year={2020},
  publisher={Elsevier}
}

@article{M72,
  title={On the B-ternary logic function-A ternary logic considering ambiguity},
  author={Mukaidono, Masao},
  journal={IEICE Trans. Inf. \& Syst.},
  volume={55},
  number={6},
  pages={355--362},
  year={1972}
}

@book{DK14,
  title={Theory of Computational Complexity},
  author={Du, Ding-Zhu and Ko, Ker-I},
  year={2014},
  publisher={John Wiley \& Sons}
}

@book{AB09,
  title={Computational complexity: a modern approach},
  author={Arora, Sanjeev and Barak, Boaz},
  year={2009},
  publisher={Cambridge University Press}
}

@article{IkenmeyerKLLMS19,
  author    = {Christian Ikenmeyer and
               Balagopal Komarath and
               Christoph Lenzen and
               Vladimir Lysikov and
               Andrey Mokhov and
               Karteek Sreenivasaiah},
  title     = {On the Complexity of Hazard-free Circuits},
  journal   = {J. {ACM}},
  volume    = {66},
  number    = {4},
  pages     = {25:1--25:20},
  year      = {2019},
  url       = {https://doi.org/10.1145/3320123},
  doi       = {10.1145/3320123},
  bibsource = {dblp computer science bibliography, https://dblp.org}
}

@inproceedings{BundLM22,
  author    = {Johannes Bund and
               Christoph Lenzen and
               Moti Medina},
  editor    = {Mark Braverman},
  title     = {Small Hazard-Free Transducers},
  booktitle = {13th Innovations in Theoretical Computer Science Conference, {ITCS}
               2022, January 31 - February 3, 2022, Berkeley, CA, {USA}},
  series    = {LIPIcs},
  volume    = {215},
  pages     = {32:1--32:24},
  publisher = {Schloss Dagstuhl - Leibniz-Zentrum f{\"{u}}r Informatik},
  year      = {2022},
  url       = {https://doi.org/10.4230/LIPIcs.ITCS.2022.32},
  doi       = {10.4230/LIPIcs.ITCS.2022.32},
  bibsource = {dblp computer science bibliography, https://dblp.org}
}

@book{Sipser97,
  author    = {Michael Sipser},
  title     = {Introduction to the theory of computation},
  publisher = {{PWS} Publishing Company},
  year      = {1997},
  isbn      = {978-0-534-94728-6},
  bibsource = {dblp computer science bibliography, https://dblp.org}
}

@book{Kleene52,
  title={Introduction to metamathematics},
  author={Kleene, Stephen Cole},
  year={1952},
  publisher={North Holland}
}

@InProceedings{IkenmeyerKS23,
  author =	{Ikenmeyer, Christian and Komarath, Balagopal and Saurabh, Nitin},
  title =	{{Karchmer-Wigderson Games for Hazard-Free Computation}},
  booktitle =	{14th Innovations in Theoretical Computer Science Conference (ITCS 2023)},
  pages =	{74:1--74:25},
  series =	{Leibniz International Proceedings in Informatics (LIPIcs)},
  ISBN =	{978-3-95977-263-1},
  ISSN =	{1868-8969},
  year =	{2023},
  volume =	{251},
  editor =	{Tauman Kalai, Yael},
  publisher =	{Schloss Dagstuhl -- Leibniz-Zentrum f{\"u}r Informatik},
  address =	{Dagstuhl, Germany},
  URL =		{https://drops.dagstuhl.de/opus/volltexte/2023/17577},
  URN =		{urn:nbn:de:0030-drops-175775},
  doi =		{10.4230/LIPIcs.ITCS.2023.74}
}

@article{Marino81,
  author    = {Leonard R. Marino},
  title     = {General Theory of Metastable Operation},
  journal   = {{IEEE} Trans. Computers},
  volume    = {30},
  number    = {2},
  pages     = {107--115},
  year      = {1981},
  url       = {https://doi.org/10.1109/TC.1981.6312173},
  doi       = {10.1109/TC.1981.6312173},
  bibsource = {dblp computer science bibliography, https://dblp.org}
}

@Article{huffman57design,
  author	= {Huffman, David A.},
  title		= {{The Design and Use of Hazard-Free Switching Networks}},
  journal	= {Journal of the ACM},
  volume	= {4},
  number	= {1},
  year		= {1957},
  pages		= {47--62},
}

@article{unger1995hazards,
  author       = {Stephen H. Unger},
  title        = {Hazards, Critical Races, and Metastability},
  journal      = {{IEEE} Trans. Computers},
  volume       = {44},
  number       = {6},
  pages        = {754--768},
  year         = {1995},
  url          = {https://doi.org/10.1109/12.391185},
  doi          = {10.1109/12.391185},
  bibsource    = {dblp computer science bibliography, https://dblp.org}
}

@article{friedrichs18containing,
  author    = {Stephan Friedrichs and
               Matthias F{\"{u}}gger and
               Christoph Lenzen},
  title     = {{Metastability-Containing Circuits}},
  journal   = {IEEE Transactions on Computers},
  volume    = {67},
  issue		= {8},
  year      = {2018}
}

@article{bund2019optimal,
  author       = {Johannes Bund and
                  Christoph Lenzen and
                  Moti Medina},
  title        = {Optimal Metastability-Containing Sorting via Parallel Prefix Computation},
  journal      = {{IEEE} Trans. Computers},
  volume       = {69},
  number       = {2},
  pages        = {198--211},
  year         = {2020},
  url          = {https://doi.org/10.1109/TC.2019.2939818},
  doi          = {10.1109/TC.2019.2939818},
  bibsource    = {dblp computer science bibliography, https://dblp.org}
}

@book{HopcroftU79,
  author    = {John E. Hopcroft and
               Jeffrey D. Ullman},
  title     = {Introduction to Automata Theory, Languages and Computation},
  publisher = {Addison-Wesley},
  year      = {1979},
  isbn      = {0-201-02988-X},
  bibsource = {dblp computer science bibliography, https://dblp.org}
}

@article{BundLM20,
  author       = {Johannes Bund and
                  Christoph Lenzen and
                  Moti Medina},
  title        = {Optimal Metastability-Containing Sorting via Parallel Prefix Computation},
  journal      = {{IEEE} Trans. Computers},
  volume       = {69},
  number       = {2},
  pages        = {198--211},
  year         = {2020},
  url          = {https://doi.org/10.1109/TC.2019.2939818},
  doi          = {10.1109/TC.2019.2939818},
  bibsource    = {dblp computer science bibliography, https://dblp.org}
}

@article{BundFLM20,
  author       = {Johannes Bund and
                  Matthias F{\"{u}}gger and
                  Christoph Lenzen and
                  Moti Medina},
  title        = {Synchronizer-Free Digital Link Controller},
  journal      = {{IEEE} Trans. Circuits Syst.},
  volume       = {67-I},
  number       = {10},
  pages        = {3562--3573},
  year         = {2020},
  url          = {https://doi.org/10.1109/TCSI.2020.2989552},
  doi          = {10.1109/TCSI.2020.2989552},
  bibsource    = {dblp computer science bibliography, https://dblp.org}
}

@inproceedings{BundFLMR20,
  author       = {Johannes Bund and
                  Matthias F{\"{u}}gger and
                  Christoph Lenzen and
                  Moti Medina and
                  Will Rosenbaum},
  title        = {{PALS:} Plesiochronous and Locally Synchronous Systems},
  booktitle    = {26th {IEEE} International Symposium on Asynchronous Circuits and Systems,
                  {ASYNC} 2020, Salt Lake City, UT, USA, May 17-20, 2020},
  pages        = {36--43},
  publisher    = {{IEEE}},
  year         = {2020},
  url          = {https://doi.org/10.1109/ASYNC49171.2020.00013},
  doi          = {10.1109/ASYNC49171.2020.00013},
  bibsource    = {dblp computer science bibliography, https://dblp.org}
}

@inproceedings{TarawnehFL17,
  author       = {Ghaith Tarawneh and
                  Matthias F{\"{u}}gger and
                  Christoph Lenzen},
  title        = {Metastability Tolerant Computing},
  booktitle    = {23rd {IEEE} International Symposium on Asynchronous Circuits and Systems,
                  {ASYNC} 2017, San Diego, CA, USA, May 21-24, 2017},
  pages        = {25--32},
  publisher    = {{IEEE} Computer Society},
  year         = {2017},
  url          = {https://doi.org/10.1109/ASYNC.2017.9},
  doi          = {10.1109/ASYNC.2017.9},
  bibsource    = {dblp computer science bibliography, https://dblp.org}
}

@Article{goto49relay,
  author	= {M. Goto},
  title		= {Application of Logical Mathematics to the Theory of Relay
		  Networks (in {J}apanese)},
  journal	= {J. Inst. Elec. Eng. of Japan},
  volume	= {64},
  number	= {726},
  year		= {1949},
  pages		= {125--130}
}

@inproceedings{FuggerKLP17,
  author       = {Matthias F{\"{u}}gger and
                  Attila Kinali and
                  Christoph Lenzen and
                  Thomas Polzer},
  title        = {Metastability-Aware Memory-Efficient Time-to-Digital Converters},
  booktitle    = {23rd {IEEE} International Symposium on Asynchronous Circuits and Systems,
                  {ASYNC} 2017, San Diego, CA, USA, May 21-24, 2017},
  pages        = {49--56},
  publisher    = {{IEEE} Computer Society},
  year         = {2017},
  url          = {https://doi.org/10.1109/ASYNC.2017.12},
  doi          = {10.1109/ASYNC.2017.12},
  bibsource    = {dblp computer science bibliography, https://dblp.org}
}

@inproceedings{FriedrichsK17,
  author       = {Stephan Friedrichs and
                  Attila Kinali},
  title        = {Efficient Metastability-Containing Multiplexers},
  booktitle    = {2017 {IEEE} Computer Society Annual Symposium on VLSI, {ISVLSI} 2017,
                  Bochum, Germany, July 3-5, 2017},
  pages        = {332--337},
  publisher    = {{IEEE} Computer Society},
  year         = {2017},
  url          = {https://doi.org/10.1109/ISVLSI.2017.65},
  doi          = {10.1109/ISVLSI.2017.65},
  bibsource    = {dblp computer science bibliography, https://dblp.org}
}

@article{Jukna21,
  author       = {Stasys Jukna},
  title        = {Notes on Hazard-Free Circuits},
  journal      = {{SIAM} J. Discret. Math.},
  volume       = {35},
  number       = {2},
  pages        = {770--787},
  year         = {2021},
  url          = {https://doi.org/10.1137/20M1355240},
  doi          = {10.1137/20M1355240},
  bibsource    = {dblp computer science bibliography, https://dblp.org}
}

@misc{Williams12,
  author        = {Ryan Williams and Luca Trevisan},
  title         = {CS154: Automata and Complexity - Notes on Rice’s Theorem},
  day           = {7},
  month         = {February},
  year          = {2012},
  publisher={Stanford University}
}

@ARTICLE{ginosar11tutorial,
author={R. Ginosar},
journal={IEEE Design Test of Computers},
title={{Metastability and Synchronizers: A Tutorial}},
year={2011},
volume={28},
number={5},
pages={23--35},
}

@Article{hu12complexity,
  author	= {W. Hu and J. Oberg and A. Irturk and M. Tiwari and T.
		  Sherwood and D. Mu and R. Kastner},
  journal	= {IEEE Transactions on Information Forensics and Security},
  title		= {{On the Complexity of Generating Gate Level Information
		  Flow Tracking Logic}},
  year		= {2012},
  volume	= {7},
  number	= {3},
  pages		= {1067--1080},
}

@Article{tiwari09flow,
  author	= {Tiwari, Mohit and Wassel, Hassan M.G. and Mazloom, Bita
		  and Mysore, Shashidhar and Chong, Frederic T. and Sherwood,
		  Timothy},
  title		= {{Complete Information Flow Tracking from the Gates Up}},
  journal	= {SIGARCH Comput. Archit. News},
  issue_date	= {March 2009},
  volume	= {37},
  number	= {1},
  year		= {2009},
  pages		= {109--120},
}

@article{BundFM23,
  author       = {Johannes Bund and
                  Matthias F{\"{u}}gger and
                  Moti Medina},
  title        = {{PALS:} Distributed Gradient Clocking on Chip},
  journal      = {{IEEE} Trans. Very Large Scale Integr. Syst.},
  volume       = {31},
  number       = {11},
  pages        = {1740--1753},
  year         = {2023},
  url          = {https://doi.org/10.1109/TVLSI.2023.3311178},
  doi          = {10.1109/TVLSI.2023.3311178},
  bibsource    = {dblp computer science bibliography, https://dblp.org}
}

@inproceedings{BundLM17,
  author       = {Johannes Bund and
                  Christoph Lenzen and
                  Moti Medina},
  editor       = {David Atienza and
                  Giorgio Di Natale},
  title        = {Near-optimal metastability-containing sorting networks},
  booktitle    = {Design, Automation {\&} Test in Europe Conference {\&} Exhibition,
                  {DATE} 2017, Lausanne, Switzerland, March 27-31, 2017},
  pages        = {226--231},
  publisher    = {{IEEE}},
  year         = {2017},
  url          = {https://doi.org/10.23919/DATE.2017.7926987},
  doi          = {10.23919/DATE.2017.7926987},
  bibsource    = {dblp computer science bibliography, https://dblp.org}
}

@article{KomarathS20,
  author       = {Balagopal Komarath and
                  Nitin Saurabh},
  title        = {On the complexity of detecting hazards},
  journal      = {Inf. Process. Lett.},
  volume       = {162},
  pages        = {105980},
  year         = {2020},
  url          = {https://doi.org/10.1016/j.ipl.2020.105980},
  doi          = {10.1016/J.IPL.2020.105980},
  bibsource    = {dblp computer science bibliography, https://dblp.org}
}

@phdthesis{bund2022hazard,
    author = {Bund, Johannes},
    title = {Hazard-free clock synchronization},
    school = {Saarl{\"a}ndische Universit{\"a}ts-und Landesbibliothek},
    year = {2022}
}

\ifnum\arxiv=0
\begin{IEEEbiography}[{\includegraphics[width=1in,height=1.25in,clip,keepaspectratio,trim=5 0 4 0]{./JB.jpg}}]{Johannes~Bund}
is a post-doc researcher at LMF, ENS Paris-Saclay since 2024.
When writing the article he was with the Faculty of Engineering at Bar-Ilan University.
He graduated with his M.\,Sc.\ studies in 2018 at the Saarland Informatics Campus and
  Max-Planck Institute for Informatics.
In 2018 he joined Christoph Lenzen's group at Max-Planck Institute for Informatics as 
  a Ph.\,D.\ student.
In 2021 he joined CISPA Helmholtz Center for Information
  Security, where he finished his Ph.\,D.\ studies in 2022.
Johannes joined Moti Medina at Bar-Ilan University as a post-doc researcher in 2023 before moving to ENS Paris-Saclay.
\end{IEEEbiography}

\begin{IEEEbiography}
[{\includegraphics[width=1in,height=1.25in,clip,keepaspectratio]{./AL.jpg}}]{Amir~Leshem}
(IEEE Fellow) received his B.Sc. (Cum Laude) in mathematics and physics, his M.Sc. (Cum Laude) in mathematics, and his Ph.D. in mathematics all from the Hebrew University, Jerusalem, Israel, in 1986, 1990 and 1998 respectively. From 1998 to 2000 he was postdoctoral researcher at Delft University of Technology, The Netherlands. From 2000 to 2003 he was the director of advanced technologies at Metalink. In 2002, he joined Bar-Ilan University, where he is a full professor.  Prof. Leshem was an associate editor of IEEE Trans. on Signal Processing ( 2008-2011) and IEEE Trans. on Signal and Information Processing over networks (2017-2021). He is an IEEE Fellow for contributions to multi-channel and multi-agent signal processing. 
His main research interests include multi-agent learning over networks, opinion dynamics, applications of game theory to networks, wireless networks, array and statistical signal processing, logic, and the foundations of mathematics.
\end{IEEEbiography}

\begin{IEEEbiography}[{\includegraphics[width=1in,height=1.25in,clip,keepaspectratio]{./MM.jpg}}]{Moti Medina}
is a faculty member in the engineering faculty at Bar-Ilan University since 2021.
Previously he was a faculty member at the School of Electrical \& Computer Engineering at
the Ben-Gurion University of the Negev since 2017. Previously, he was a post-doc
researcher in MPI for Informatics and in the Algorithms and Complexity group at
LIAFA (Paris 7). He graduated with his Ph.\,D., M.\,Sc., and B.\,Sc.\ studies at the
School of Electrical Engineering at Tel-Aviv University, in  2014, 2009, and 2007
respectively. Moti is also a co-author of a  text-book on logic design
``Digital Logic Design: A Rigorous Approach'', Cambridge Univ. Press, 2012.
\end{IEEEbiography}
\fi
\ifnum\supp=0
\ifnum\arxiv=0
\appendices
\else
\appendix
\fi

\section{Kleene Logic and metastability-containing combinational circuits}\label{sec:kleenecomb}
To define the metastability-containing Turing machine, one needs to be able to construct metastability-containing circuits as building blocks. In this appendix, we show how to represent the effect of metastability on combinational circuits using the Kleene ternary logic. For completeness, we also include some basic definitions and a discussion on a classical combinational circuit - the multiplexer.
\color{black}
\subsection{Combinational Circuits}
A \emph{Combinational Circuit} $C$ has $n$ input gates and $m$ output gates and is constructed with $\AND$, $\OR$, and $\NOT$ gates, where $\AND$ and $\OR$ have fan-in of $2$, $\NOT$ gates have a fan-in of $1$, and input gates have fan-in $0$. The out-degree can be any number except the output gates with an out-degree of $0$. A combinational circuit cannot have any cycles, i.e., a combinational circuit is a directed acyclic graph. Each gate computes a Boolean function: an input gate outputs the corresponding input bit, an $\AND$ gate that is fed by input gates that output $x_i, x_j \in \IB$ outputs $\AND(x_i, x_j)$ (similarly for $\OR$ and $\NOT$ gates), and $\AND$ gate that is fed by gates $f_1, f_2 \in \{\AND, \OR, \NOT\}$ that output $y_1, y_2$, respectively, output $\AND(y_1, y_2)$. The values $y_1$ and $y_2$ are computed recursively. Finally, each output gate of $C$ computes the Boolean function computed by the gate that feeds it. The \emph{size} of a combinational circuit $C$ is the number of non-input and output gates and is denoted by $|C|$. The \emph{depth} of a $C$ is the longest path in $C$ from an input gate to an output gate.

From the definition above, every combinational circuit computes a Boolean function $f: \IB^n\rightarrow \IB^m$. On the other hand, it is well known that a combinational circuit can implement every Boolean function $f: \IB^n\rightarrow \IB^m$. As such, combinational circuits can implement \emph{any} transition and output function of a given \TM\ that operates on inputs encoded by bits. 

\subsection{Kleene's ternary logic}\label{sec:kleenelogic}
Kleene’s classical three valued \emph{strong logic of indeterminacy}~\cite[\S64]{Kleene52} captures the issues arising from non-digital
inputs. The two-valued Boolean logic $\IB$ is extended by a third value $\mfu$ representing any unknown, uncertain, undefined, transitioning, or non-binary value. We call both Boolean values \emph{stable}, while $\mfu$ is called \emph{unstable}. Let $\IT=\{0,1,\mfu\}$ denote this set of extended values. 
\begin{figure}
    \centering
    \begin{tikzpicture}[x=.5cm,y=.8cm]
        \node at (0,0) (u) {$\mfu$};
        \node at (-1,-1) (0) {$0$};
        \node at (1,-1) (1) {$1$};
    
        \draw (u) -- (0);
        \draw (u) -- (1);
    \end{tikzpicture}
    \caption{The lattice for $\preceq$.}
    \label{fig:lattice}
\end{figure}
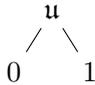

The partial order $\preceq$ over $\IT$ is defined as follows~\cite{M72}: $\mfu \preceq 0$, $\mfu \preceq 1$ and $i\preceq i$ for $i\in\IT$, i.e., $0$ and $1$ do not relate to each other, while $\mfu$ is the minimal w.r.t.\ $\preceq$. Intuitively, $0$ and $1$ are more stable than $\mfu$ or a possible stabilization of $\mfu$. This relation is extended to $\IT^n$, as follows: for $x,y \in \IT^n$, $x \preceq y$ iff for all $i \in [n]$, $x_i \preceq y_i$. 
The lattice for $\preceq$ is depicted in Figure~\ref{fig:lattice}.

In Kleene's ternary logic, the basic operands $\AND$, $\OR$, and $\NOT$ (maybe together with the constants $\{0,1\}$) are extended to their \MC versions, which we denote by $\AND_\mfu, \OR_\mfu$ and $\NOT_\mfu$ (see Table~1 for the full truth tables). Note that, $\AND_\mfu, \OR_\mfu$ and $\NOT_\mfu$ are monotone w.r.t.\ to $\preceq$.
Now, with these extended gates, a combinational circuit implements (following a similar recursive evaluation procedure as in the Boolean case, as illustrated above) a function $f: \IT^n \rightarrow \IT$. Moreover, due to the monotonicity w.r.t.\ to $\preceq$ of $\AND_\mfu, \OR_\mfu$ and $\NOT_\mfu$, any function implemented by a combinational circuit is monotone w.r.t.\ to $\preceq$, that is, unlike the Boolean case, combinational circuits \emph{do not} implement \emph{all} the functions in $\IT^n \rightarrow \IT$~\cite{M72}.
%
\begin{theorem}{\cite[Thm.~3]{M72}}\label{thm:natural}
    A function $f : \IT^n \rightarrow \IT$ can be implemented by a combinational circuit iff
    \begin{itemize}
        \item $f \vert_{\IB^n} : \IB^n \rightarrow \IB$, and 
        \item $f$ is monotone w.r.t.\ $\preceq$,
    \end{itemize}
    we refer to such Boolean functions as  \emph{natural} functions.
\end{theorem}

\subsection{Implementability of the \MC\ of Boolean Functions by Family of Combinational Circuits}\label{sec:mcimple}
We now argue that the \MC\ of Boolean functions with a single output, i.e., $f\colon\IB^*\rightarrow\IB$, is implementable by families of combinational circuits. In turn, this allows defining a subclass of hardware implementable \TM\, i.e., implemented using gates, wires, and memory registers. The implementation of these gates is discussed above. 

Naturally, considering inputs of length $n$ and restricting $f_\mfu$ to binary inputs, yields a Boolean function: $f_\mfu \vert_{\IB^n} = f \vert_{\IB^n} \colon\IB^n\rightarrow\IB$.
Moreover, for a given input length, the \MC\ of $f$, $f_\mfu$, is monotone w.r.t.\ $\preceq$, as follows. Consider $x \preceq y$, then $\res(y) \subseteq \res(x)$. Hence, $f(\res(y)) \subseteq f(\res(x))$. Monotonicity of $f_\mfu$ follows from the definition of the superposition (see Definition\ifnum\supp=1 ~\textbf{III.3}\else ~\ref{def:closure}\fi). 
Hence, by Theorem~\ref{thm:natural}, the \MC of $f\colon\IB^*\rightarrow\IB$, can be implemented by a family of combinational circuits, i.e.,  the \MC of $f$ restricted to inputs of length $n$ can be implemented by a circuit for each $n\in\IN$. 

Given a Boolean function $f$, its \MC $f_{\mfu}$ is unique and maximal w.r.t.\ $\preceq$~\cite{friedrichs18containing, IkenmeyerKLLMS19}, as such, combinational circuits that implement their \MC are desirable since their output is the most stable one possible.~\footnote{In the literature circuits that compute the \MC are sometimes referred to to as 
\emph{hazard-free} circuits}
%

\paragraph{M-Containing Multiplexer\ifnum\arxiv=1 .\fi}
A commonly used tool in the study of M-Containing circuits is the multiplexer.
The $n$ to $1$ multiplexer is a circuit that selects $1$ out of $n$ bits 
    according to a select input.
The metastability-containing multiplexer is a multiplexer that has a correct output for any input in $\IT^n$. Common implementations for basic multiplexers and extension (by recursive implementation) to  $\NCMUX$ are presented, e.g., in~\cite{friedrichs18containing,IkenmeyerKLLMS19,IkenmeyerKS23}.
\ifnum\supp=0
\begin{definition}[M-Containing Multiplexer (\CMUX)]\label{def:cmux}
    Given $\ell\in\IN_{>0}$, such that $n=2^\ell$;
        an $n$ to $1$ \emph{multiplexer} is given by the Boolean function  
        $\NMUX\colon\IB^n\times\IB^\ell\rightarrow\IB$. 
    Let $x\in\IB^n$ be a string of input bits and $s\in\IB^\ell$ be the select input, then
    \begin{align*}
        \NMUX(x,s)\coloneqq x_p \text{, and } p \coloneqq {\langle s\rangle_B}\,,
    \end{align*}
    where $\langle s \rangle_B \in \IN$ denotes the number encoded (in binary encoding) by $s$,
    and $x_p$ denotes the bit at position $p$ in $x$.
    
    The $n$ to $1$ \emph{M-Containing Multiplexer ($\CMUX$)} is given by the function 
        $\NCMUX\colon\IT^n\times\IT^\ell\rightarrow\IT$.
    It is defined by the M-Containing of the $n$ to $1$ multiplexer;
    \begin{align*}
        \NCMUX(x,s) \coloneqq \left(\NMUX(x,s)\right)_\mfu\,,
    \end{align*}
    where $x\in\IT^n$ and $s\in\IT^\ell$.
\end{definition}
\fi

\fi
\vfill

\end{document}